\documentclass[12pt]{article}

\usepackage[english]{babel}

\usepackage[margin=1in]{geometry}  
\usepackage{setspace}              
\usepackage{times}                 
\usepackage[authoryear]{natbib}

\usepackage{graphicx}
\usepackage{float}
\usepackage{verbatim}
\usepackage{multirow}
\usepackage{booktabs} 
\usepackage{threeparttable}
\usepackage{amsmath,amssymb, amsbsy}
\usepackage{mathtools}
\usepackage{tikz-cd}
\usepackage{bbm}
\usepackage{tikz}
\usepackage{changepage}
\usepackage{makecell}
\usepackage{algorithm,algpseudocode}
\usepackage{xcolor}
\definecolor{evergreen}{RGB}{0,102,76}

\usepackage{amsthm}
\usepackage[shortlabels]{enumitem}
\usepackage[colorlinks=true, allcolors=blue]{hyperref}
\newtheoremstyle{roman}
  {\topsep}   
  {\topsep}   
  {\normalfont} 
  {}          
  {\normalfont\bfseries} 
  {.}         
  { }         
  {}          
\newtheorem{thm}{Theorem}
\newtheorem{cond}{Condition}

\newtheorem{lemma}{Lemma}
\newtheorem{prop}{Proposition}
\newtheorem{coro}{Corollary}

\newtheorem{assumption}{Assumption} 

\theoremstyle{remark}

\theoremstyle{definition}
\newtheorem{definition}{Definition}

\theoremstyle{remark}
\newtheorem*{remark}{Remark}

\DeclareMathOperator*{\ep}{\mathbb{E}}

\DeclareMathOperator*{\ind}{\mathbbm{1}}

\newcommand{\indep}{\perp \!\!\! \perp}

\title{Conformalized Lee Inference: Distribution-Free Prediction Sets for Individual Treatment Effect under Monotone Sample Selection}
\author{Jung Hyub Lee \footnote{Graduate School of Economics, University of Tokyo, 7-3-1 Hongo, Bunkyo-ku, Tokyo, Japan 113-0033, Email: \href{mailto:jhlee@e.u-tokyo.ac.jp}{jhlee@e.u-tokyo.ac.jp}.}  }
\date{\today}
\begin{document}

\maketitle
\vspace{-2em}
\begin{center}
\end{center}
\begin{abstract}
Treatment can change which outcomes are observed, so treated-selected and selected-control units need not represent the same latent population. This paper proposes conformalized Lee inference for counterfactual prediction under randomized treatment and monotone sample selection. The procedure uses treated-selected observations to fit and calibrate an arbitrary prediction rule and replaces the usual $(1-\alpha)$ score quantile with the adjusted $(1-\alpha\pi)$ quantile, where $\pi$ is the identified share of always-selected units among treated-selected units. The resulting prediction set has finite-sample, distribution-free marginal coverage over the sharp Lee ambiguity set. For a selected-control unit, subtracting the observed untreated outcome yields a marginal prediction set for the realized individual treatment effect. The adjusted population cutoff is minimax optimal over the reduced-information identification region. Simulations show that ordinary split-conformal prediction can under-cover under distribution shifts induced by selection, whereas the adjusted procedures restore coverage. Empirical analysis uses the National Job Corps Study data to illustrate prediction sets for individual wage effects of assignment to program access.
\end{abstract}



\section{Introduction}

Sample-selection problems become more difficult when treatment changes whether an outcome is observed. In a randomized study, treated-selected and selected-control units can then represent different principal strata, so comparisons of observed outcomes combine treatment effects with changes in sample composition. The monotone-selection condition of \citet{lee2009training} rules out units who would be selected under control but not under treatment. Under this condition, selected controls are always-selected units, whereas treated-selected units comprise always-selected and marginal-in units.

This paper proposes \emph{conformalized Lee inference}, a distribution-free method for prediction under monotone sample selection. The primary target is a prediction set for the treated potential outcome of an independent draw from the always-selected population. For a selected-control unit, the untreated outcome is observed and monotonicity reveals membership in the always-selected stratum. Subtracting the observed untreated outcome from the treated-outcome prediction set therefore gives a prediction set for the realized individual treatment effect. The guarantee is marginal over independent selected-control draws; it is not conditional on a fixed unit's covariates and untreated outcome.

The calibration distribution differs from the target distribution. Let $P$ denote the observed treated-selected distribution of $(X,Y(1))$, and let $Q_0$ denote the corresponding always-selected distribution. If $p_d=\Pr\{S(d)=1\}$ for $d\in\{0,1\}$, random assignment identifies $p_0$ and $p_1$, and monotonicity gives
$\pi=p_0/p_1$, the share of always-selected units among treated-selected units. Because principal-stratum membership is latent among treated-selected units, $Q_0$ is not point identified.

The first main result characterizes this distributional identification problem. The always-selected distribution is dominated by the treated-selected distribution, and its likelihood ratio is bounded by $1/\pi$:
\[
Q_0 \ll P,
\qquad
0 \leq \frac{dQ_0}{dP} \leq \frac{1}{\pi}.
\]
This restriction extends Lee's trimming argument from the outcome distribution to the joint distribution of covariates and treated potential outcomes. The paper also establishes sharpness relative to the reduced observables $(P,p_0,p_1)$. Every distribution satisfying the likelihood ratio restriction can arise under random assignment, monotone selection, and the same observed selection rates.

The second main result converts the identification restriction into a conformal prediction rule. Ordinary split-conformal prediction controls the score tail under the calibration distribution $P$ but not under the target distribution $Q_0$. The likelihood ratio bound implies that a target distribution can assign at most $1/\pi$ times the $P$ probability to any high-score event. Thus, the proposed procedure uses the $(1-\alpha\pi)$ treated-selected score quantile instead of the usual $(1-\alpha)$ quantile. This choice makes the treated-selected tail probability no larger than $\alpha\pi$ and hence makes the target tail probability no larger than $\alpha$.

The procedure fits any prediction rule on treated-selected training observations, computes nonconformity scores on a treated-selected calibration fold, and applies the adjusted cutoff. The coverage result is finite-sample and distribution-free for arbitrary measurable prediction rules and score functions. When $\pi$ is unknown, a value no larger than the true share preserves the coverage guarantee but widens the prediction set. The paper shows that one-sided lower confidence bounds provide high-probability guarantees, whereas the plug-in estimator gives a less conservative implementation with a one-sided error bound.

The third main result gives a minimax interpretation of the cutoff. Over the sharp Lee ambiguity set, the worst-case target distribution places as much probability as permitted on the upper score tail. Consequently, the $(1-\alpha\pi)$ population quantile is the smallest threshold that guarantees coverage uniformly over the reduced-information identification region. The paper establishes that every smaller threshold fails for some observationally equivalent data-generating process.

The Monte Carlo study compares naive conformal prediction with oracle, plug-in, and lower bound Lee adjustments under benign, tail-selection, and smooth-selection designs. Naive conformal prediction is valid when the calibration and target distributions coincide, but it under-covers when always-selected units are more likely to have large nonconformity scores. The adjusted procedures restore coverage in the shifted designs. In contrast, the Clopper--Pearson and Hoeffding implementations suggested in the paper are conservative because they replace $\pi$ with a lower confidence bound.


The paper use the National Job Corps Study data analyzed by \cite{lee2009training} to illustrate prediction sets for individual wage effects of assignment to program access. Ridge and quantile regression models use baseline characteristics and assignment timing to predict log hourly wages among selected-treated units. For each selected-control unit, we evaluate the calibrated wage interval at that person’s covariates and subtract the observed-control log wage from both endpoints. The result shows that most individual effect sets span both wage gains and wage losses under either score. 

The remainder of the paper is organized as follows. Section \ref{sec:setup} defines the principal strata and the predictive target, derives the Lee ambiguity set, and establishes sharp identification from the reduced observables. Section \ref{sec:counterfactual_prediction} introduces the split-conformal procedure, proves finite-sample counterfactual and ITE coverage, and characterizes the minimax cutoff. Section \ref{sec:simulation} reports the Monte Carlo evidence. Section \ref{sec:empirical} presents the Job Corps illustration and distinguishes its numerical results from the coverage guarantees of the model. The appendix contains the proofs, practical lower bound constructions for $\pi$, and the identification result based on the full observed distribution.
\subsection{Related Work}


\paragraph{Sample selection and Lee bounds.}
The sample-selection literature studies inference on treatment effects when outcomes are observed only for selected units (see, among others, \citet{lee2009training, honore2020selection, honore2024sample}). \citet{lee2009training} shows that monotone selection yields sharp bounds on average treatment effects by trimming the treated-selected outcome distribution. The present paper retains the same principal-strata restriction but changes the inferential target. It constructs distribution-free prediction sets for treated counterfactual outcomes and realized individual treatment effects in the always-selected population.

Recent work extends Lee-type corrections to covariate rich selection models, generalized Lee bounds, continuous treatments, random objects, and endogenous treatment (see, e.g., \citet{heiler2024heterogeneous, semenova2025generalized, lee2024lee, kurisu2026lee, dong2026sharp}). These papers primarily bound average or distributional causal parameters. The present paper instead uses Lee-type nesting to identify an ambiguity set for the target counterfactual distribution and then converts that set into a conformal calibration rule. The full observed distribution result in the appendix is close to covariate rich Lee methods. It identifies the always-selected covariate marginal and imposes the conditional restriction
\[
Q_x \ll P_x,
\qquad
0 \leq \frac{dQ_x}{dP_x} \leq \frac{1}{\pi(x)}.
\]
The main procedure deliberately uses only reduced observables $(P,p_0,p_1)$, which produces a simple distribution-free rule that may be conservative relative to the full observed distribution.

\paragraph{Principal stratification and partial identification.}
Principal stratification defines causal effects within latent groups determined by post-treatment variables (see, e.g., \cite{frangakis2002principal}). The truncated-outcome problem in \citet{zhang2003estimation} has the same basic structure, with survival replacing sample selection. Under monotone selection, selected controls reveal the always-selected stratum, whereas treated-selected observations combine always-selected and marginal-in units. In the present paper, this stratum structure determines the causal target and the distributional restriction used for conformal calibration.

The identification argument also follows the logic of partial identification and contamination models (see, among others, \citet{huber1981robust, horowitz1995identification, manski2003partial}). The observed treated-selected distribution satisfies
\[
P=\pi Q+(1-\pi)R,
\]
where $Q$ is the always-selected target distribution and $R$ is an unrestricted marginal-in distribution. The latent identity of always-selected treated units prevents point identification of $Q$. Random assignment and monotone selection identify the mixture share $\pi=p_0/p_1$, which yields the sharp likelihood ratio ambiguity set. This mixture representation is the distributional counterpart of Lee's trimming argument.

\paragraph{Conformal prediction under distribution shift.}
The prediction argument builds on conformal inference under exchangeability and split-conformal regression (see, e.g., \citet{vovk2005algorithmic, lei2018distribution}). Standard split-conformal prediction is not directly valid for the proposed method because the calibration scores follow the treated-selected distribution $P$, whereas the counterfactual target follows an unknown $Q$ in the Lee ambiguity set.

Weighted conformal prediction addresses known covariate shifts, and more general methods allow departures from exchangeability (see, e.g., \citet{tibshirani2019conformal, barber2023conformal}). The present setting differs in two respects. First, the pointwise density ratio is not identified since monotone selection identifies only the sharp joint bound $dQ/dP\leq 1/\pi$. Second, principal-stratum membership may depend on both $X$ and $Y(1)$, so the calibration and target distributions can differ in their covariate marginals and conditional outcome distributions. Therefore, the conformalized Lee cutoff  provides robust calibration for a partially identified shift in the joint distribution of $(X,Y(1))$.

\paragraph{Conformal causal inference and sensitivity analysis.}
The paper also bridges conformal methods to counterfactuals and individual treatment effects (see, e.g., \citet{lei2021conformal, chernozhukov2021exact}). The closest comparison is robust conformal sensitivity analysis (see, e.g., \citet{jin2023sensitivity, yin2024conformal}), together with related sensitivity methods for unobserved confounding and departures from identifying assumptions (see, e.g., \citet{rosenbaum1987sensitivity, yadlowsky2022bounds, dorn2023sharp, dorn2025doubly}).

In those approaches, the researcher specifies a sensitivity class. In contrast, the monotonicity and the observed selection rates in the paper identify the robustness class. They also determine the closed-form change from the usual $(1-\alpha)$ cutoff to the $(1-\alpha\pi)$ treated-selected score quantile. The minimax result shows that every smaller population cutoff fails for some target distribution in the reduced-information identification region.

\section{Setup, Principal Strata, and Predictive Target}
\label{sec:setup}

We use the potential outcome framework and consistency convention introduced by \citet{imbens2015causal}. Consider a randomized treatment setting with sample selection. For unit $i=1,\ldots,n$, let $D_i\in\{0,1\}$ denote treatment assignment, $X_i\in\mathcal X\subseteq\mathbb R^d$ denote predetermined covariates, and $Y_i(1)$ and $Y_i(0)$ denote the potential outcomes. Let $S_i(1),S_i(0)\in\{0,1\}$ denote the potential selection indicators. Define the potential outcome vector $W_i=(Y_i(0),Y_i(1),S_i(0),S_i(1),X_i)$ and suppose that $W_1,\ldots,W_n$ are independent and identically distributed. The realized selection indicator is $S_i=D_iS_i(1)+(1-D_i)S_i(0)$, and $Y_i=D_iY_i(1)+(1-D_i)Y_i(0)$ is observed when $S_i=1$. The observed sample is $\{(D_i,S_i,S_iY_i, X_i)\}_{i=1}^n$.

Define the assignment vector $\mathbf D_n = (D_1, \ldots, D_n)$ and the potential outcome vector $\mathbf W_n = (W_1, \ldots, W_n)$. We maintain the following assumptions throughout the paper.

\begin{assumption}[Random assignment and assignment positivity]
\label{ass:random}
\[
\mathbf D_n \indep \mathbf W_n
\quad \text{and} \quad 
\rho_i = \Pr(D_i=1)\in(0,1) \,\, for \,\, i=1,\ldots,n.
\]
\end{assumption}
\begin{assumption}[Monotone selection]
\label{ass:monosel}
\[
S(1)\geq S(0)\quad\text{almost surely.}
\]
\end{assumption}


Assumption \ref{ass:random} relates the observed distribution in treatment status $d$ to the corresponding distribution of the potential variables under $d$. Vector-level assignment independence further implies the unit-level independence $D_i \indep W_i$ for every $i$. Also, the marginal assignment probabilities $\rho_i$ may differ across units.

Assumption \ref{ass:monosel} rules out units who would be selected under control but not under treatment. The treated-selected sample generally do not represent the always-selected target population. Under monotone selection, treatment can induce additional units to enter the observed sample. Thus, treated-selected observations contain always-selected and marginal-in units, whereas selected-control observations contain only always-selected units. Hence, the population can be categorized into three principal strata:
\[
\mathcal A=\{S(0)=S(1)=1\},
\qquad
\mathcal M=\{S(0)=0,S(1)=1\},
\qquad
\mathcal N=\{S(0)=S(1)=0\},
\]
where $\mathcal A$ is the always-selected stratum, $\mathcal M$ is the marginal-in stratum, and $\mathcal N$ is the never-selected stratum. Every selected-control unit belongs to $\mathcal A$, whereas a treated-selected unit can belong to $\mathcal A$ or $\mathcal M$.

The predictive target is the individual treatment effect in the always-selected stratum:
\[
\tau_i=Y_i(1)-Y_i(0),\qquad i\in\mathcal A.
\]
For a selected-control unit, $Y(0)$ is observed and monotonicity reveals membership in $\mathcal A$. Therefore, if there is a prediction set for the missing $Y(1)$ in $\mathcal A$, it can be shifted by the observed $Y(0)$ to obtain a prediction set for $\tau_i$.

Let $p_d=\Pr\{S(d)=1\}$ for $d\in\{0,1\}$. Because the potential outcome vectors $W_1, \ldots W_n$ have a common distribution, Assumption \ref{ass:random} gives, for every unit with positive probability of assignment to status $d$,
\[
p_d=\Pr(S=1\mid D=d).
\]
Define
$\pi=p_0/p_1=\Pr\{\mathcal A\mid S(1)=1\}$. The quantity $\pi$ is the always-selected share among treated-selected units, and $1-\pi$ is the corresponding marginal-in share.

\begin{assumption}[Positivity]
\label{ass:positivity}
\[
p_0=\Pr\{S(0)=1\}>0.
\]
\end{assumption}

Assumption \ref{ass:positivity} ensures that the always-selected stratum is nonempty. Assumption \ref{ass:monosel} then gives $p_1\geq p_0>0$, so $\pi=p_0/p_1$ is well defined and belongs to $(0,1]$.
\subsection{Lee-Type Distributional Identification}

We use regular conditional distributions and measurable probability kernels on standard Borel spaces. Let $Z=(X,Y(1))\in\mathcal Z=\mathcal X\times\mathcal Y$, where $(\mathcal Z,\mathcal B)$ is a standard Borel space with $\mathcal B=\mathcal B_{\mathcal X}\otimes\mathcal B_{\mathcal Y}$. Let $\mathcal P(\mathcal Z,\mathcal B)$ denote the probability measures on this space.

The treated-selected distribution is
\[
P=\mathcal F\{Z\mid S(1)=1\}.
\]
For every unit $i$, Assumption \ref{ass:random} and consistency give
\[
\mathcal F(X,Y\mid D=1,S=1)=P.
\]
We suppress the unit index when writing the observable treated-selected distribution. The target distribution is
\[
Q_0=\mathcal F(Z\mid S(0)=1)
   =\mathcal F(Z\mid\mathcal A).
\]
The distribution $Q_0$ is not observed because principal-stratum membership is latent among treated-selected units. 

Among treated-selected units, the always-selected share is $\pi$ and the marginal-in share is $1-\pi$. Hence,
\[
P=\pi Q_0+(1-\pi)R,
\]
where
\[
R=\mathcal F\{Z\mid S(0)=0,S(1)=1\}
\]
is the marginal-in distribution. If $\pi=1$, the marginal-in stratum has probability zero and $P=Q_0$. If $\pi<1$, the unrestricted distribution $R$ prevents point identification of $Q_0$. This is a contaminated-sampling representation of the treated-selected distribution (see, \citet{huber1981robust, horowitz1995identification}).

The mixture representation extends the trimming argument of \cite{lee2009training} to the joint distribution of $(X,Y(1))$. The conformal procedure in Section \ref{sec:counterfactual_prediction} calibrates nonconformity scores under $P$, whereas its coverage target is a future always-selected draw from $Q_0$. Monotone selection gives the nesting events $\{S(0)=1\}\subseteq\{S(1)=1\}$ and therefore $Q_0$ must be dominated by $P$. Thus, $Q_0$ cannot assign positive probability to a $P$-null event, and it cannot assign more than $1/\pi$ times the $P$ probability to any measurable event.
For $P\in\mathcal P(\mathcal Z,\mathcal B)$ and $\pi\in(0,1]$, define the following ambiguity set.

\begin{definition}[Lee ambiguity set]
\label{def:lee_ambiguity}
\[
\mathcal Q(P,\pi)
=
\left\{
Q\in\mathcal P(\mathcal Z,\mathcal B):
Q\ll P,
\quad
0\leq\frac{dQ}{dP}\leq\frac{1}{\pi}
\quad P\text{-almost surely}
\right\}.
\]
\end{definition}

The density restriction in Definition \ref{def:lee_ambiguity} has the same form as a trimmed probability class (see, \citet{cascos2005integral, alvarez2008trimmed}). With trimming fraction $1 - \pi$, the retained distribution has density bounded by $1 / \pi$ relative to $P$. Its interpretation as the always-selected treated distribution follows from the monotonicity. 

Proposition \ref{prop:ident} shows that the true always-selected distribution belongs to $\mathcal Q(P,\pi)$. This identification result is the basis for prediction sets in Section \ref{sec:counterfactual_prediction}. 

\begin{prop}[Lee-type likelihood ratio domination]
\label{prop:ident}
Suppose Assumptions \ref{ass:monosel} and \ref{ass:positivity} hold. Then $Q_0\ll P$ and its Radon--Nikodym derivative satisfies
\[
0\leq\frac{dQ_0}{dP}(z)\leq\frac{1}{\pi}
\quad P\text{-almost surely}.
\]
Equivalently, $Q_0\in\mathcal Q(P,\pi)$.

If $\pi<1$, there exists $R_{\rm res}\in\mathcal P(\mathcal Z,\mathcal B)$ such that
\[
P=\pi Q_0+(1-\pi)R_{\rm res}.
\]
If $\pi=1$, then $Q_0=P$.
\end{prop}
\subsection{Sharpness of the Lee Ambiguity Set Relative to Reduced Observables}

Proposition \ref{prop:ident} gives a necessary likelihood ratio restriction using only monotone selection and selection positivity. This subsection establishes sufficiency result relative to the reduced observables $(P,p_0,p_1)$. For every $Q\in\mathcal Q(P,\pi)$, a latent data-generating process can satisfy Assumptions \ref{ass:monosel} and \ref{ass:positivity}, match the same latent reduced quantities, and assign $Q$ to the always-selected treated outcomes. Appending any assignment variable satisfying Assumption \ref{ass:random} converts the treated-selected latent distribution into the corresponding observed treated-selected distribution. Consequently, $\mathcal Q(P,\pi)$ is the exact identified region.

The construction follows directly from the mixture representation. For a candidate $\widetilde Q\in\mathcal Q(P,\pi)$ and $\pi<1$, define
\[
R(B)=\frac{P(B)-\pi\widetilde Q(B)}{1-\pi},
\qquad B\in\mathcal B.
\]
The bound $0\leq d\widetilde Q/dP\leq1/\pi$ makes $R$ a probability distribution and gives
\[
P=\pi\widetilde Q+(1-\pi)R.
\]
Assigning $\widetilde Q$ to the always-selected stratum and $R$ to the marginal-in stratum reproduces the treated-selected distribution. Because the marginal-in distribution is unrestricted, the reduced observables cannot distinguish among candidate distributions in $\mathcal Q(P,\pi)$.

Proposition \ref{prop:sharp} formalizes this argument. Any procedure that is uniformly valid under Assumptions \ref{ass:random}--\ref{ass:positivity}, and that uses only $(P,p_0,p_1)$, must therefore be valid for every distribution in $\mathcal Q(P,\pi)$. 

\begin{prop}[Sharpness of $\mathcal Q(P,\pi)$ relative to reduced observables]
\label{prop:sharp}
Let $(\mathcal Z,\mathcal B)$ be a standard Borel space. Fix $P\in\mathcal P(\mathcal Z,\mathcal B)$, $0<p_0\leq p_1\leq1$, and $\pi=p_0/p_1$. For every $\widetilde Q\in\mathcal Q(P,\pi)$, there exists a latent data-generating process for $(Y(1),Y(0),S(1),S(0),X)$ satisfying Assumptions \ref{ass:monosel} and \ref{ass:positivity} such that
\[
\Pr\{S(0)=1\}=p_0,
\qquad
\Pr\{S(1)=1\}=p_1,
\]
and
\[
\mathcal F\{X,Y(1)\mid S(1)=1\}=P,
\qquad
\mathcal F\{X,Y(1)\mid S(0)=1\}=\widetilde Q.
\]
Append any binary assignment variable $D$ that is independent of the constructed potential outcome vector and satisfies $0<\Pr(D=1)<1$. If $(S,Y)$ are defined by consistency, then
\[
\mathcal F(X,Y\mid D=1,S=1)=P.
\]
For a sample, independent copies of the latent construction may be combined with any prescribed joint distribution for $\mathbf D_n$ having positive unit-specific margins, provided that design distribution is appended independently of the latent sample. Hence, Assumption \ref{ass:random} is satisfied. Consequently, conditional on $(P,p_0,p_1)$, the sharp identification region for $Q_0$ is $\mathcal Q(P,\pi)$.
\end{prop}
\section{Counterfactual Prediction for Always-Selected Units}
\label{sec:counterfactual_prediction}

This section constructs prediction sets for the unobserved treated potential outcome $Y(1)$ and the realized ITE of always-selected units. 

Consider the score-tail event
\[
A_t=\{z\in\mathcal Z:V(z)>t\}.
\]
The likelihood ratio restriction transfers this event across the two distributions. Since every $Q\in\mathcal Q(P,\pi)$ satisfies the restriction
\[
Q(A_t)\leq\frac{1}{\pi}P(A_t),
\]
$P(A_t)\leq\alpha\pi$ implies $Q(A_t)\leq\alpha$ uniformly over the Lee ambiguity set. Conditional on the training information and random split, the calibration nonconformity scores are exchangeable draws from $P$. However, the target nonconformity score follows an unknown $Q\in\mathcal Q(P,\pi)$ and is not generally exchangeable with them. Thus, ordinary split-conformal prediction controls the score-tail under $P$, but not necessarily under $Q$. The proposed procedure does not restore exchangeability between $P$ and $Q$. Instead, it combines this transfer inequality with the usual exchangeable rank argument under $P$ and estimates the $(1-\alpha\pi)$ score quantile.

\subsection{Split-Conformal Algorithm}

Let $I_1=\{i:D_i=1,S_i=1\}$ and $I_0=\{i:D_i=0,S_i=1\}$. Then randomly split $I_1$ into a training fold $I_{\rm train}$ and a calibration fold $I_{\rm cal}$. The training fold determines a prediction rule and a measurable nonconformity score $V:\mathcal Z\to\mathbb R$. Let $\mathcal H$ be the sigma-field generated by $I_{\rm train}$, the assignment and selection sample $\{(D_i,S_i)\}_{i=1}^n$, the random split $R_{\rm sp}$, and any model-fitting randomness $R_{\rm fit}$. Then the split indices are measurable functions of $(\mathbf D_n,\mathbf S_n,R_{\rm sp})$, while the fitted score depends on $I_{\rm train}$ and $R_{\rm fit}$.

Let $m=|I_{\rm cal}|$ and $\widehat\pi\in[0,1]$ be the estimate of always-selected share among treated-selected units used for calibration. Define
\[
k=\left\lceil(m+1)(1-\alpha\widehat\pi)\right\rceil.
\]
These quantities are $\mathcal H$-measurable and are fixed conditional on $\mathcal H$. Write order statistics $V_{(1)}\leq\ldots\leq V_{(m)}$ and set $V_{(m+1)}=+\infty$. The threshold is $\widehat t=V_{(k)}$, and the prediction set for $Y(1)$ is
\[
\widehat C_1(x)=\{y:V(x,y)\leq\widehat t\}.
\]
When $\pi=1$, the target and calibration distributions coincide and the rule reduces to ordinary split-conformal prediction. When $\pi<1$, the higher $(1-\alpha\pi)$ quantile protects against the additional upper-tail mass permitted by the Lee ambiguity set. Algorithm 1 suggests the steps for computing the prediction set.

Write $I_{\rm cal} = \{i_1 < \ldots< i_m \}$ and $Z_j = (X_{i_j}, Y_{i_j})$ for $j = 1, \ldots m$. Lemma \ref{lem:calibration-design} provides the conditional product distribution required by the conformal rank argument in Theorem \ref{thm1:coverage}.

\begin{lemma}[A sufficient randomized-design condition]
\label{lem:calibration-design}
Suppose Assumption \ref{ass:random} holds and
\[
    R_{\rm sp} \indep (\mathbf D_n, \mathbf W_n), \quad R_{\rm fit} \indep (\mathbf D_n, \mathbf W_n, R_{\rm sp}).
\]
Then, for every $r \in \{0, \ldots n\}$ with $\Pr(m=r) > 0$,
\[
\mathcal F( (Z_1,\ldots,Z_r)  \mid\mathcal H)=P^{\otimes r}
\quad\text{almost surely on } \{m=r\},
\]
where $P^{\otimes 0}$ denotes the point mass on the empty tuple.
\end{lemma}

\begin{algorithm}[t]
\caption{Split-conformal Lee prediction}
\begin{algorithmic}[1]
\State \textbf{Input:} Data $\mathcal D=\{(D_i,X_i,S_i,S_iY_i)\}_{i=1}^n$, miscoverage level $\alpha\in(0,1)$, a nonconformity-score function $V: \mathcal Z \to \mathbb R$, and test covariate $x$.
\State Estimate or otherwise choose $\widehat\pi\in[0,1]$ from the assignment and selection sample $\{(D_i,S_i)\}_{i=1}^n$.
\State Form $I_1=\{i:D_i=1,S_i=1\}$ and $I_0=\{i:D_i=0,S_i=1\}$.
\State Randomly split $I_1$ into $I_{\rm train}$ and $I_{\rm cal}$.
\State Fit the prediction rule on $I_{\rm train}$ and obtain the score function $V$.
\State For each $i\in I_{\rm cal}$, compute $V_i=V(X_i,Y_i)$.
\State Let $m=|I_{\rm cal}|$, order the scores as $V_{(1)}\leq\cdots\leq V_{(m)}$, and set $V_{(m+1)}=+\infty$.
\State Compute $k=\left\lceil(m+1)(1-\alpha\widehat\pi)\right\rceil$ and set $\widehat t=V_{(k)}$.
\State \textbf{Output:} $\widehat C_1(x)=\{y:V(x,y)\leq\widehat t\}$.
\end{algorithmic}
\end{algorithm}

Let $Z_{m+1}=(X_{m+1},Y_{m+1})$ be an external target draw from $Q\in\mathcal Q(P,\pi)$ that is independent of the study sample, and  $\mathbb P$ denote the resulting joint distribution. Conditional on $\mathcal H$ and on every event $\{m=r\}$, the regular conditional distribution of the ordered calibration observations and the target draw is $P^{\otimes r}\otimes Q$ by Lemma \ref{lem:calibration-design}. Theorem \ref{thm1:coverage} establishes finite-sample marginal coverage uniformly over the Lee ambiguity set.

\begin{thm}[Finite-sample counterfactual coverage]
\label{thm1:coverage}
Let $(\mathcal Z,\mathcal B)$ be a standard Borel space, $P\in\mathcal P(\mathcal Z,\mathcal B)$, and $Q\in\mathcal Q(P,\pi)$ for some $\pi\in(0,1]$. Fix $\alpha\in(0,1)$ and let $\widehat\pi\in[0,1]$ be the value used by the algorithm. Suppose Assumption \ref{ass:random} holds and the target observation is conditionally independent of the calibration sample given $\mathcal H$, with conditional distribution $Q$. Then
\[
\mathbb P\{Y_{m+1}\in\widehat C_1(X_{m+1})\mid\mathcal H\}
\geq 1-\alpha-\widehat\Delta,
\]
where
\[
\widehat\Delta
=
\left[
\frac{m+1-k}{(m+1)\pi}-\alpha
\right]_+ \quad \text{and} \quad \widehat\Delta
\leq
\frac{\alpha}{\pi}(\widehat\pi-\pi)_+.
\]
\end{thm}

Theorem \ref{thm1:coverage} shows that the coverage error is one-sided. If $0<\widehat\pi\leq\pi$, then $\widehat\Delta=0$ because the algorithm calibrates against the larger class $\mathcal Q(P,\widehat\pi)\supseteq\mathcal Q(P,\pi)$. If $\widehat\pi>\pi$, the algorithm understates the possible upper-tail inflation, and coverage can fall below $1-\alpha$ by at most
\[
\alpha\left(\frac{\widehat\pi}{\pi}-1\right).
\]
If $\widehat\pi=0$ set $k=m+1$, $\widehat t=+\infty$, and the prediction set is the full outcome space. If $\pi=1$ and the oracle value $\widehat\pi=1$ is used, the procedure reduces to ordinary split-conformal prediction. The appendix gives Clopper--Pearson and Hoeffding lower bound constructions for $\pi$.

The calibration value also determines informativeness. As $\widehat\pi$ decreases, the index $\lceil(m+1)(1-\alpha\widehat\pi)\rceil$ increases and the prediction set becomes weakly wider. At $\widehat\pi=0$, the threshold is $V_{(m+1)}=+\infty$, so the prediction set is uninformative. Hence, a useful lower bound should be as large as possible while remaining no greater than the true $\pi$.

Proof of Theorem \ref{thm1:coverage} separates the distribution-shift step from the conformal-rank step. The likelihood ratio restriction transfers a score-tail event from $Q$ to $P$. Then, an auxiliary draw $Z_{m+1}^{\circ}\sim P$ represents the random $P$ tail probability as an exchangeable event. Independent uniform tie breakers make the $m$ calibration scores and the auxiliary score almost surely distinct. Thus, their rank is uniform conditional on $\mathcal H$. Combining the transfer inequality with this rank bound yields Theorem \ref{thm1:coverage}.


The finite-sample result does not require consistent estimation of the
prediction rule or of the score distribution. Nevertheless, a consistent estimate of $\pi$ is sufficient to recover the oracle coverage target asymptotically. For Corollary \ref{coro:plugin-asymptotic-coverage}, we introduce the subscript $n$ to emphasize that the quantities and results depend on the sample size $n$. Consider a sequence of procedures indexed by $n$.  Let $\mathcal H_n$ denote the conditioning sigma-field in Algorithm 1, and let $m_n$, $V_n$,
$\widehat\pi_n$, $k_n$, $\widehat t_n$, and $\widehat C_{1,n}$ denote the corresponding calibration size, trained score, calibration value, order statistic index, threshold, and prediction set, respectively. For the uniformity statement in Corollary \ref{coro:plugin-asymptotic-coverage}, only the independent target distribution $Q$ varies over $\mathcal Q(P,\pi)$. Consequently, the distribution of $\widehat\pi_n$ and the coverage error bound is common across these target distributions. 

\begin{coro}[Asymptotic validity with an estimated share]
\label{coro:plugin-asymptotic-coverage}
Conditional on $\mathcal
H_n$, suppose that the calibration observations are independent draws from
$P$ and that the test observation is an independent draw from
$Q\in\mathcal Q(P,\pi)$.  Let $\mathbb P_{n,Q}$ denote the resulting joint distribution.

If $\widehat\pi_n\overset{p}{\longrightarrow}\pi$, then, for every $Q\in\mathcal Q(P,\pi)$,
\[
\left[
1-\alpha-
\mathbb P_{n,Q}\!
\left\{
Y_{n,m_n+1}\in
\widehat C_{1,n}(X_{n,m_n+1})
\mid\mathcal H_n
\right\}
\right]_+
\leq
\frac{\alpha}{\pi}(\widehat\pi_n-\pi)_+
=o_p(1).
\]
Moreover,
\[
\liminf_{n\to\infty}
\inf_{Q\in\mathcal Q(P,\pi)}
\mathbb P_{n,Q}\!
\left\{
Y_{n,m_n+1}\in
\widehat C_{1,n}(X_{n,m_n+1})
\right\}
\geq 1-\alpha.
\]
\end{coro}

Note that Corollary \ref{coro:plugin-asymptotic-coverage} does not require $m_n \rightarrow \infty$. Although the one-sided coverage error vanishes when $\widehat \pi_n$ is consistent, $\widehat t_n$ itself may remain random.

Let $\Gamma=\mathcal F\{X,Y(0),Y(1)\mid\mathcal A\}$ and let $Q_\Gamma=\mathcal F\{X,Y(1)\mid\mathcal A\}$ be its $(X,Y(1))$ marginal distribution. Conditional on $\mathcal H$, define $\mathbb P_\Gamma$ so that the calibration observations follow $P^{\otimes m}$ and an independent test unit $(X_{m+1},Y_{m+1}(0),Y_{m+1}(1))$ follows $\Gamma$. Assumptions \ref{ass:monosel} and \ref{ass:positivity}, through Proposition \ref{prop:ident}, imply $Q_\Gamma\in\mathcal Q(P,\pi)$. Therefore, the prediction set for $Y(1)$ can be shifted to form an ITE prediction set.

\begin{definition}[ITE prediction set]
For $(x,y_0)\in \mathcal Z$, 
\[
\widehat C_\tau(x,y_0)
=
\widehat C_1(x)-y_0
=
\{y-y_0:y\in\widehat C_1(x)\}.
\]
\end{definition}

\citet{lei2021conformal} introduce this type of prediction target for the counterfactual and ITEs. Corollary \ref{coro:iteband} converts marginal coverage for $Y(1)$ into marginal coverage for the realized ITE of an independent selected-control draw.

\begin{coro}[Marginal ITE prediction for selected controls]
\label{coro:iteband}
Let
\[
(X_{m+1},Y_{m+1}(0),Y_{m+1}(1))\sim\Gamma
\]
be a test unit independent of $\mathcal H$ and the calibration observations, and define
\[
\tau_{m+1}=Y_{m+1}(1)-Y_{m+1}(0).
\]
Then
\[
\mathbb P_\Gamma
\left\{
\tau_{m+1}\in\widehat C_\tau\bigl(X_{m+1},Y_{m+1}(0)\bigr)
\mid\mathcal H
\right\}
\geq1-\alpha-\widehat\Delta.
\]
\end{coro}
\subsection{Population Optimality and Threshold Consistency}
\label{subsec:optimality}


Theorem \ref{thm1:coverage} establishes finite-sample validity of the threshold $\widehat t=V_{(k)}$. This subsection shows that the adjustment from $1-\alpha$ to $1-\alpha\pi$ also gives the sharp population cutoff implied by the reduced-information identification region.

Fix a trained score function $V$ and consider $C_t=\{z:V(z)\leq t\}$. Uniform population coverage requires
\[
\inf_{Q\in\mathcal Q(P,\pi)}Q(V\leq t)\geq1-\alpha,
\]
or, equivalently,
\[
\sup_{Q\in\mathcal Q(P,\pi)}Q(V>t)\leq\alpha.
\]
For every $Q\in\mathcal Q(P,\pi)$, the likelihood ratio restriction gives
\[
Q(V>t)\leq\frac{1}{\pi}P(V>t).
\]
Therefore, the population problem is a minimax calibration problem over the upper score tail.

Proposition \ref{prop:sharpscore} shows that the upper bound is attained for every threshold. In consequence, the minimax coverage condition is equivalent to
\[
P(V>t)\leq\alpha\pi.
\]
The smallest admissible threshold is the lower $(1-\alpha\pi)$ quantile of the score under $P$. 

\begin{prop}[Sharp tail bound over the Lee ambiguity set]
\label{prop:sharpscore}
Let $(\mathcal Z,\mathcal B)$ be a standard Borel space, $P\in\mathcal P(\mathcal Z,\mathcal B)$, and $\pi\in(0,1]$. For every measurable score $V:\mathcal Z\to\mathbb R$ and every $t\in\mathbb R$,
\[
\sup_{Q\in\mathcal Q(P,\pi)}Q(V>t)
=
\min\left\{1,\frac{P(V>t)}{\pi}\right\}.
\]
For each $t$, a distribution $Q_t\in\mathcal Q(P,\pi)$ attains this supremum.
\end{prop}

Let $F(t)=P(V\leq t)$ and
$F^{-1}(u)=\inf\{t\in\overline{\mathbb R}:F(t)\geq u\}$, where $\overline{\mathbb R}=\mathbb R\cup\{+\infty\}$. Define the population minimax threshold as follows.
\begin{definition}[Population minimax threshold]
\[
t^*
=
\inf\left\{
t\in\overline{\mathbb R}:
\inf_{Q\in\mathcal Q(P,\pi)}Q(V\leq t)\geq1-\alpha
\right\}.
\]
\end{definition}

Proposition \ref{prop:sharpthreshold} gives an explicit expression for $t^*$. Among prediction sets defined by a score cutoff, $t^*$ is the smallest threshold that is uniformly valid over the sharp reduced-information region. The split-conformal threshold $\widehat t$ in Theorem \ref{thm1:coverage} is its finite-sample, distribution-free analogue.

\begin{prop}[Minimax population threshold]
\label{prop:sharpthreshold}
Let $(\mathcal Z,\mathcal B)$ be a standard Borel space. Fix $P\in\mathcal P(\mathcal Z,\mathcal B)$, $\pi\in(0,1]$, and $\alpha\in(0,1)$. For every measurable score $V:\mathcal Z\to\mathbb R$, 
\[
t^*=F^{-1}(1-\alpha\pi).
\]
Equivalently, $t^*$ is the smallest threshold satisfying $P(V>t^*)\leq\alpha\pi$, and $P(V>t)>\alpha\pi$ for every $t<t^*$. Moreover,
\[
\inf_{Q\in\mathcal Q(P,\pi)}Q(V\leq t)
=
\max\left\{0,1-\frac{P(V>t)}{\pi}\right\}
\]
for every $t\in\mathbb R$, and some $Q_t\in\mathcal Q(P,\pi)$ attains the infimum. Hence, no threshold below $t^*$ is uniformly valid.
\end{prop}



The result is closely related to robust conformal sensitivity analysis studied by \citet{jin2023sensitivity}. Their probably approximately correct (PAC) procedure constructs a conservative lower envelope $\widehat G_n$ for the worst-case score distribution and chooses a threshold at $\inf \{ t: \widehat G_n(t) \ge 1 - \alpha \}$. However, the Lee ambiguity set is identified in the present setting. Monotone selection and the selection rates give the likelihood ratio restriction $0\leq dQ/dP\leq1/\pi$, which implies the closed-form worst-case distribution function
\[
G_{\rm Lee}(t)
=
\inf_{Q\in\mathcal Q(P,\pi)}Q(V\leq t)
=
\max\left\{0,1-\frac{P(V>t)}{\pi}\right\}.
\]
The conformalized Lee procedure exploits this scalar upper-tail form and uses order statistics of $V$ to derive $t^*$. In consequence, $t^*$ is a closed-form threshold obtained from the identified Lee ambiguity set. 

Proposition \ref{prop:sharpthreshold} concerns the population cutoff $t^*$ for a fixed trained score, whereas Theorem \ref{thm1:coverage} concerns a random finite-sample order statistic $\widehat t$. Proposition \ref{prop:threshold-consistency} bridges these two objects.  The first conclusion compares the sample cutoff with the population cutoff induced by the realized trained score.  A deterministic limiting cutoff requires the additional stability condition stated in the second part.

\begin{prop}[Consistency of the split-conformal threshold]
\label{prop:threshold-consistency}
Let $\mathcal H_n$ be the conditioning sigma-field, $m_n$ be the calibration size, and $V_n:\mathcal Z\to\mathbb R$ be the $\mathcal H_n$-measurable trained score.  Conditional on $\mathcal H_n$, suppose that
\[
V_{n,i}=V_n(Z_{n,i}),\qquad i=1,\ldots,m_n,
\]
are independent with common conditional CDF
\[
F_n(t)
=
\int \mathbbm 1\{V_n(z)\leq t\}\,P(dz).
\]
Fix $\alpha\in(0,1)$ and $\pi\in(0,1]$, and define
\[
q=1-\alpha\pi,
\qquad
t_n^*=F_n^{-1}(q)
=
\inf\{t\in\mathbb R:F_n(t)\geq q\}.
\]
Let $\widehat\pi_n\in[0,1]$ be $\mathcal H_n$-measurable and set
\[
k_n
=
\left\lceil(m_n+1)(1-\alpha\widehat\pi_n)\right\rceil,
\qquad
\widehat t_n=V_{n,(k_n)},
\]
with $V_{n,(m_n+1)}=+\infty$.  Suppose that
\[
m_n \overset{p}{\longrightarrow} \infty
\qquad\text{and}\qquad
\widehat\pi_n \overset{p}{\longrightarrow} \pi.
\]
Then, $\Pr(k_n=m_n+1)\to0$.  In addition, each of the following conclusions
holds under its stated regularity condition.

\begin{enumerate}[(i)]
\item Suppose that the conditional $q$-quantile is asymptotically isolated in the sense that
for every $\varepsilon>0$, there is an $\eta_\varepsilon>0$ such that
\[
\Pr\!\left\{
F_n(t_n^*-\varepsilon)\leq q-\eta_\varepsilon,
\quad
F_n(t_n^*+\varepsilon)\geq q+\eta_\varepsilon
\right\}
\longrightarrow1.
\]
Then
\[
\widehat t_n-t_n^*\overset{p}{\longrightarrow}0.
\]

\item Suppose that there is a nonrandom CDF $F$ such that
\[
\sup_{t\in\mathbb R}|F_n(t)-F(t)|\overset{p}{\longrightarrow}0
\]
and $t^*=F^{-1}(q)$ is isolated in the sense that for every $\epsilon > 0$,
\[
F(t^*-\varepsilon)<q<F(t^*+\varepsilon)
.
\]
Then
\[
t_n^* \overset{p}{\longrightarrow} t^*
\qquad\text{and}\qquad
\widehat t_n \overset{p}{\longrightarrow}t^*.
\]
\end{enumerate}
\end{prop}

Proposition \ref{prop:plugin-pi-consistency} in the \hyperref[sec:appendix]{Appendix} provides the required conditions for $\widehat \pi_n \overset{p}{\longrightarrow} \pi$ and $m_n \overset{p}{\longrightarrow} \infty$.

\section{Simulation Studies}
\label{sec:simulation}

\begin{table}[t]
\centering
\caption{Monte Carlo design for the focused simulation study}
\label{tab:mc-design}
\begin{threeparttable}
\small
\begin{tabular}{ll}
\toprule
Component & Value \\
\midrule
Dimension & $p=4$ \\
Target calibration size & $m\in\{100,200\}$ \\
Always-selected share & $\pi\in\{0.75,0.50,0.25\}$ \\
Target coverage grid & $1-\alpha\in\{0.50,0.60,0.70,0.80,0.90\}$ \\
Treatment probability & $\Pr(D=1)=1/2$ \\
Treated selection rate & $p_1=\Pr\{S(1)=1\}=0.8$ \\
Outcome design & Heteroskedastic linear model for $Y(1)$ \\
DGPs & Benign, conditional tail, unconditional tail, smooth \\
Nonconformity scores & Oracle residual, fitted residual, CQR \\
Methods & Naive, oracle Lee, plug-in Lee, CP-Lee, Hoeffding-Lee \\
Monte Carlo replications & $R=100$ per configuration \\
CI budget for lower bound methods & $\delta_\pi=0.05$ \\
\bottomrule
\end{tabular}
\begin{tablenotes}[flushleft]
\footnotesize
\item  \textit{Notes:} CP-Lee and Hoeffding-Lee use the same nominal miscoverage $\alpha$ in the conformal cutoff as the other methods. Their formal unconditional coverage target is $1-\alpha-\delta_\pi$.
\end{tablenotes}
\end{threeparttable}
\end{table}

\begin{table}[t]
\centering
\caption{Principal-stratum assignment rules in the simulation DGPs}
\label{tab:dgp-rules}
\begin{threeparttable}
\small
\setlength{\tabcolsep}{4pt}
\begin{tabular}{p{0.22\linewidth}p{0.54\linewidth}p{0.17\linewidth}}
\toprule
DGP & Always-selected rule among units with $S(1)=1$ & Purpose \\
\midrule
Benign & $A\sim\rm{Bernoulli}(\pi)$ independently of $(X,U_1,Y(1))$ & No distribution shift \\
Conditional tail & $A=1\{|U_1|>\sigma(X)\Phi^{-1}(1-\pi/2)\}$ & Near-worst-case score shift \\
Unconditional tail & $A=1\{|U_1|>c_\pi\}$, where $\Pr(|U_1|>c_\pi)=\pi$ & Joint outcome and covariate shift \\
Smooth & $A\mid X,U_1\sim\rm{Bernoulli}\{\rm{logit}^{-1}(a_\pi+b_X(2X_1-1)+b_Y(|U_1|/\sigma(X)-q_{0.75}))\}$ & Smooth heterogeneous selection \\
\bottomrule
\end{tabular}
\begin{tablenotes}[flushleft]
\footnotesize
\item  \textit{Notes:} Each design sets $S(0)=\ind\{A=1\}$, so $S(1)\geq S(0)$ holds by construction. The smooth design uses $(b_X,b_Y)=(1,1.5)$.
\end{tablenotes}
\end{threeparttable}
\end{table}

\begin{table}[t]
\centering
\caption{Coverage summary over the full focused grid}
\label{tab:coverage-full-grid}
\begin{threeparttable}
\small
\begin{tabular}{llrrrrrr}
\toprule
DGP & Method for $\pi_{\rm used}$ & Coverage & Gap & Min. gap & Below & Finite length  \\
\midrule
Benign & Naive & 0.703 & 0.003 & -0.008 & 20/90 & 3.145  \\
    & Oracle Lee & 0.854 & 0.154 & 0.022 & 0/90 & 4.564  \\
     & Plug-in Lee & 0.853 & 0.153 & 0.023 & 0/90 & 4.569 \\
     & CP-Lee & 0.878 & 0.228 & 0.084 & 0/90 & 4.948 \\
     & Hoeffding-Lee & 0.894 & 0.244 & 0.088 & 0/90 & 5.097 \\
     \\
Conditional tail & Naive & 0.427 & -0.273 & -0.602 & 90/90 & 3.144 \\
     & Oracle Lee & 0.714 & 0.014 & -0.016 & 6/90 & 4.567 \\
     & Plug-in Lee & 0.714 & 0.014 & -0.005 & 6/90 & 4.574 \\
     & CP-Lee & 0.767 & 0.117 & 0.063 & 0/90 & 4.951 \\
     & Hoeffding-Lee & 0.804 & 0.154 & 0.069 & 0/90 & 5.109 \\
     \\
Unconditional tail & Naive & 0.416 & -0.284 & -0.698 & 90/90 & 3.178 \\
     & Oracle Lee & 0.711 & 0.011 & -0.018 & 6/90 & 4.598 \\
     & Plug-in Lee & 0.712 & 0.012 & -0.011 & 8/90 & 4.610 \\
     & CP-Lee & 0.765 & 0.115 & 0.061 & 0/90 & 4.973 \\
     & Hoeffding-Lee & 0.803 & 0.153 & 0.072 & 0/90 & 5.082 \\
     \\
Smooth & Naive & 0.565 & -0.135 & -0.282 & 90/90 & 3.175 \\
     & Oracle Lee & 0.766 & 0.066 & 0.005 & 0/90 & 4.592 \\
     & Plug-in Lee & 0.767 & 0.067 & 0.006 & 0/90 & 4.598 \\
     & CP-Lee & 0.804 & 0.154 & 0.068 & 0/90 & 4.972 \\
     & Hoeffding-Lee & 0.833 & 0.183 & 0.072 & 0/90 & 5.138 \\
\bottomrule
\end{tabular}
\begin{tablenotes}[flushleft]
\footnotesize
\item  \textit{Notes:} ``Gap'' is empirical coverage minus the method-specific target. The target is $1-\alpha$ for naive, oracle Lee, and plug-in Lee, and $1-\alpha-\delta_\pi$ for CP-Lee and Hoeffding-Lee. ``Below'' is the number of summarized configurations with a negative gap. ``Finite length'' averages only configurations with finite reported lengths.
\end{tablenotes}
\end{threeparttable}
\end{table}

\begin{table}[t]
\centering
\caption{Coverage and length at nominal coverage $1-\alpha=0.90$}
\label{tab:coverage-alpha01}
\begin{threeparttable}
\small
\begin{tabular}{llrrrrr}
\toprule
DGP & Method for $\pi_{\rm used}$ & Target & Coverage & Gap & Finite length \\
\midrule
Benign & Naive & 0.900 & 0.903 & 0.003 & 4.863 \\
     & Oracle Lee & 0.900 & 0.953 & 0.053 & 6.217 \\
     & Plug-in Lee & 0.900 & 0.954 & 0.054 & 6.254  \\
     & CP-Lee & 0.850 & 0.962 & 0.112 & 6.657 \\
     & Hoeffding-Lee & 0.850 & 0.967 & 0.117 & 6.569 \\
     \\
Conditional tail & Naive & 0.900 & 0.771 & -0.129 & 4.863 \\
     & Oracle Lee & 0.900 & 0.909 & 0.009 & 6.212 \\
     & Plug-in Lee & 0.900 & 0.909 & 0.009 & 6.237 \\
     & CP-Lee & 0.850 & 0.928 & 0.078 & 6.668 \\
     & Hoeffding-Lee & 0.850 & 0.940 & 0.090 & 6.519 \\
     \\
Unconditional tail & Naive & 0.900 & 0.769 & -0.131 & 4.915 \\
     & Oracle Lee & 0.900 & 0.910 & 0.010 & 6.273 \\
     & Plug-in Lee & 0.900 & 0.912 & 0.012 & 6.311 \\
     & CP-Lee & 0.850 & 0.929 & 0.079 & 6.699 \\
     & Hoeffding-Lee & 0.850 & 0.940 & 0.090 & 6.579 \\
     \\
Smooth & Naive & 0.900 & 0.821 & -0.079 & 4.913 \\
     & Oracle Lee & 0.900 & 0.919 & 0.019 & 6.278 \\
     & Plug-in Lee & 0.900 & 0.919 & 0.019 & 6.266  \\
     & CP-Lee & 0.850 & 0.935 & 0.085 & 6.701  \\
     & Hoeffding-Lee & 0.850 & 0.945 & 0.095 & 6.572  \\
\bottomrule
\end{tabular}
\begin{tablenotes}[flushleft]
\footnotesize
\item  \textit{Notes:} The target is $0.90$ for naive, oracle Lee, and plug-in Lee. Under the same-$\alpha$ convention, the target is $0.85$ for CP-Lee and Hoeffding-Lee because $\delta_\pi=0.05$. The table averages over $m\in\{100,200\}$, $\pi\in\{0.25,0.50,0.75\}$, and the three score classes.
\end{tablenotes}
\end{threeparttable}
\end{table}

\begin{table}[t]
\centering
\caption{Always-selected share diagnostics at nominal coverage $1-\alpha=0.90$}
\label{tab:eta-diagnostics}
\begin{threeparttable}
\small
\begin{tabular}{rrrrrrrr}
\toprule
$m$ & $\pi$ & Avg. $m$ & $E[\widehat\pi]$ & $E[\widehat \pi_{\rm{CP}}]$ & $E[\widehat\pi_{\rm{H}}]$ & CP $\Pr(k=m+1)$ & H $\Pr(k=m+1)$ \\
\midrule
100 & 0.25 & 100.0 & 0.251 & 0.181 & 0.120 & 0.000 & 0.242 \\
100 & 0.50 & 100.0 & 0.499 & 0.399 & 0.342 & 0.000 & 0.000 \\
100 & 0.75 & 100.5 & 0.749 & 0.631 & 0.565 & 0.000 & 0.000 \\
200 & 0.25 & 200.8 & 0.251 & 0.200 & 0.156 & 0.000 & 0.000 \\
200 & 0.50 & 199.9 & 0.501 & 0.428 & 0.386 & 0.000 & 0.000 \\
200 & 0.75 & 200.0 & 0.751 & 0.667 & 0.617 & 0.000 & 0.000 \\
\bottomrule
\end{tabular}
\begin{tablenotes}[flushleft]
\footnotesize
\item  \textit{Notes:} The table averages over the DGPs and score classes at $\alpha=0.1$. The plug-in estimator is close to the true $\pi$, whereas the lower confidence bounds are deliberately smaller. The last two columns report the probability that $k=m+1$, which yields an infinite prediction interval.
\end{tablenotes}
\end{threeparttable}
\end{table}

\begin{table}[t]
\centering
\caption{Score robustness under selection-induced distribution shift at nominal coverage $1-\alpha=0.90$}
\label{tab:score-robustness}
\begin{threeparttable}
\small
\begin{tabular}{llrrrrr}
\toprule
Score & Method for $\pi_{\rm used}$& Target & Coverage & Gap & Finite length & $\Pr(k=m+1)$ \\
\midrule
Oracle residual & Naive & 0.900 & 0.782 & -0.118 & 4.775 & 0.000 \\
     & Oracle Lee & 0.900 & 0.911 & 0.011 & 6.114 & 0.000 \\
     & Plug-in Lee & 0.900 & 0.913 & 0.013 & 6.161 & 0.000 \\
    & CP-Lee & 0.850 & 0.930 & 0.080 & 6.584 & 0.000 \\
    & Hoeffding-Lee & 0.850 & 0.942 & 0.092 & 6.469 & 0.039 \\
    \\
Fitted residual & Naive & 0.900 & 0.783 & -0.117 & 4.863 & 0.000 \\
     & Oracle Lee & 0.900 & 0.911 & 0.011 & 6.235 & 0.000 \\
     & Plug-in Lee & 0.900 & 0.913 & 0.013 & 6.272 & 0.000 \\
     & CP-Lee & 0.850 & 0.930 & 0.080 & 6.683 & 0.000 \\
     & Hoeffding-Lee & 0.850 & 0.941 & 0.091 & 6.552 & 0.045 \\
     \\
CQR & Naive & 0.900 & 0.794 & -0.106 & 5.052 & 0.000 \\
     & Oracle Lee & 0.900 & 0.917 & 0.017 & 6.415 & 0.000 \\
     & Plug-in Lee & 0.900 & 0.915 & 0.015 & 6.382 & 0.000 \\
     & CP-Lee & 0.850 & 0.932 & 0.082 & 6.801 & 0.000 \\
     & Hoeffding-Lee & 0.850 & 0.943 & 0.093 & 6.650 & 0.042 \\
\bottomrule
\end{tabular}
\begin{tablenotes}[flushleft]
\footnotesize
\item  \textit{Notes:} The table averages over the conditional-tail, unconditional-tail, and smooth designs and excludes the benign design. It compares the Lee correction across oracle residual, fitted residual, and CQR scores.
\end{tablenotes}
\end{threeparttable}
\end{table}

The Monte Carlo study examines the coverage mechanism of Algorithm 1 under four selection designs. The benign design provides a no-shift benchmark, while the conditional-tail, unconditional-tail, and smooth designs assume different shifts between the treated-selected calibration distribution and the always-selected target distribution.

Table \ref{tab:mc-design} summarizes the design. In each replication, an independent sample from the always-selected population is used to evaluate coverage and interval length. The potential outcomes follow a heteroskedastic linear model. We draw
\[
X\sim\operatorname{Unif}([0,1]^4)
\]
and
\[
U_1\mid X\sim N\{0,\sigma^2(X)\},
\qquad
\sigma^2(X)=1+\frac{1}{2}(2.5X_1)^2.
\]
The treated potential outcome is
\[
Y(1)=\beta^\top X+U_1,
\qquad
\beta=(-0.531,0.126,-0.312,0.018)^\top.
\]
For ITE evaluation, the untreated potential outcome is
\[
Y(0)=\beta_0^\top X+0.5U_1+U_0,
\qquad
U_0\sim N(0,1),
\qquad
\beta_0=(0.2,-0.1,0.1,0.05)^\top.
\]
The shared disturbance $U_1$ induces dependence between $Y(1)$ and $Y(0)$.

Table \ref{tab:dgp-rules} defines the principal-stratum assignments. In every design,
\[
D\sim\operatorname{Bernoulli}(1/2),
\qquad
S(1)\sim\operatorname{Bernoulli}(0.8),
\]
and the always-selected indicator $A$ is assigned only among units with $S(1)=1$. We set $S(0)=\ind\{A=1\}$, so monotonicity holds by construction. The benign design makes $A$ independent of $(X,U_1,Y(1))$, which gives $P=Q$. The conditional-tail design selects units from the conditional upper tail of $|U_1|$. The unconditional-tail design selects from the unconditional upper tail and also changes the covariate distribution because $\sigma(X)$ depends on $X_1$. The smooth design uses a logistic selection probability based on $X_1$ and $|U_1|/\sigma(X)$.

The simulations compare five calibration rules. Naive conformal sets $\pi_{\rm used}=1$, oracle Lee uses the true $\pi$, and plug-in Lee sets  $\pi_{\rm used}=\min\{1,\widehat p_0/\widehat p_1\}$. CP-Lee and Hoeffding-Lee use one-sided lower confidence bounds. The lower bound methods use the same conformal miscoverage $\alpha$ as the other procedures and allocate $\delta_\pi=0.05$ to estimation of $\pi$. Therefore, their formal unconditional target is $1-\alpha-\delta_\pi$.

The study uses three nonconformity scores. The oracle residual score is
\[
V(x,y)=|y-\beta^\top x|,
\]
which isolates calibration from model estimation and prediction. The fitted residual score is
\[
V(x,y)=|y-\widehat\mu_1(x)|,
\]
where $\widehat\mu_1$ is estimated on the treated-selected training fold. The CQR score  (\citet{romano2019conformalized}) is
\[
V(x,y)
=
\max\{\widehat q_{\alpha/2}(x)-y,\;y-\widehat q_{1-\alpha/2}(x)\},
\]
where the conditional quantile functions are estimated on the same training fold (see, \citet{meinshausen2006quantile}). These scores compare the oracle calibration experiment with two feasible prediction methods.

Table \ref{tab:coverage-full-grid} reports coverage over the full grid, and the gap column compares each method with its method-specific target. Under benign selection, naive conformal prediction has an average gap of $0.003$, as expected when $P=Q$. Under distribution shift, its average gap falls to $-0.273$ in the conditional-tail design, $-0.284$ in the unconditional-tail design, and $-0.135$ in the smooth design. The larger failures in the two tail designs reflect their concentration of always-selected units in the upper score tail.

Oracle Lee and plug-in Lee remove most of these shortfalls. Their average gaps are $0.014$ in the conditional-tail design, approximately $0.01$ in the unconditional-tail design, and approximately $0.07$ in the smooth design. CP-Lee and Hoeffding-Lee are more conservative because their lower bounds use smaller values of $\pi_{\rm used}$. The resulting higher score quantiles raise coverage and interval length, with the Hoeffding procedure giving the largest gaps at these sample sizes.

Table \ref{tab:coverage-alpha01} focuses on nominal coverage $0.90$. Naive conformal prediction covers $0.903$ in the benign design but only $0.771$, $0.769$, and $0.821$ in the conditional-tail, unconditional-tail, and smooth designs. Oracle and plug-in Lee cover between $0.909$ and $0.919$ in the shifted designs. CP-Lee and Hoeffding-Lee exceed their formal target of $0.85$ under the same-$\alpha$ convention.

Table \ref{tab:eta-diagnostics} explains the conservativeness of the lower bound methods. The plug-in estimator is nearly unbiased. When $\pi=0.25$, its mean is $0.251$ for both calibration sizes. With $m=100$, the corresponding CP and Hoeffding lower bounds average $0.181$ and $0.120$. The smaller Hoeffding value also explains the main infinite-interval case. When $m=100$ and $\pi=0.25$, $\Pr(k=m+1)=0.242$. Because such intervals are valid but uninformative, the tables report finite length and the probability of $k=m+1$ separately.

Table \ref{tab:score-robustness} averages over the three shifted designs at nominal coverage $0.90$. Naive conformal prediction covers $0.782$, $0.783$, and $0.794$ with the oracle residual, fitted residual, and CQR scores. Oracle and plug-in Lee cover about $0.91$ for all three scores, while the lower bound methods are conservative relative to their formal target. Hence, the proposed correction remains effective when the score depends on estimated prediction functions.

Figures \ref{fig:dgp-benign}--\ref{fig:dgp-smooth} display the same ranking across the coverage grid. Naive conformal prediction lies near the diagonal in the benign design but falls below it in each shifted design, with the largest shortfalls in the two tail designs and at smaller $\pi$. Plug-in Lee remains close to the target under conditional-tail, unconditional-tail, and smooth selection. CP-Lee and Hoeffding-Lee generally lie above the target, and the Hoeffding curves are highest because that lower bound is more conservative. 

The tables and figures support a common conclusion. Naive conformal prediction is accurate when $P=Q$ but can fail under selection-induced shifts in the joint distribution of $(X,Y(1))$. Plug-in Lee closely tracks the oracle benchmark, whereas the lower bound procedures trade wider and occasionally infinite sets for conservative coverage. Increasing $m$ from $100$ to $200$ mainly reduces finite-sample variability. 

\begin{figure}[p]
\centering
\caption{Coverage under benign selection}
\includegraphics[width=1 \linewidth]
{\detokenize{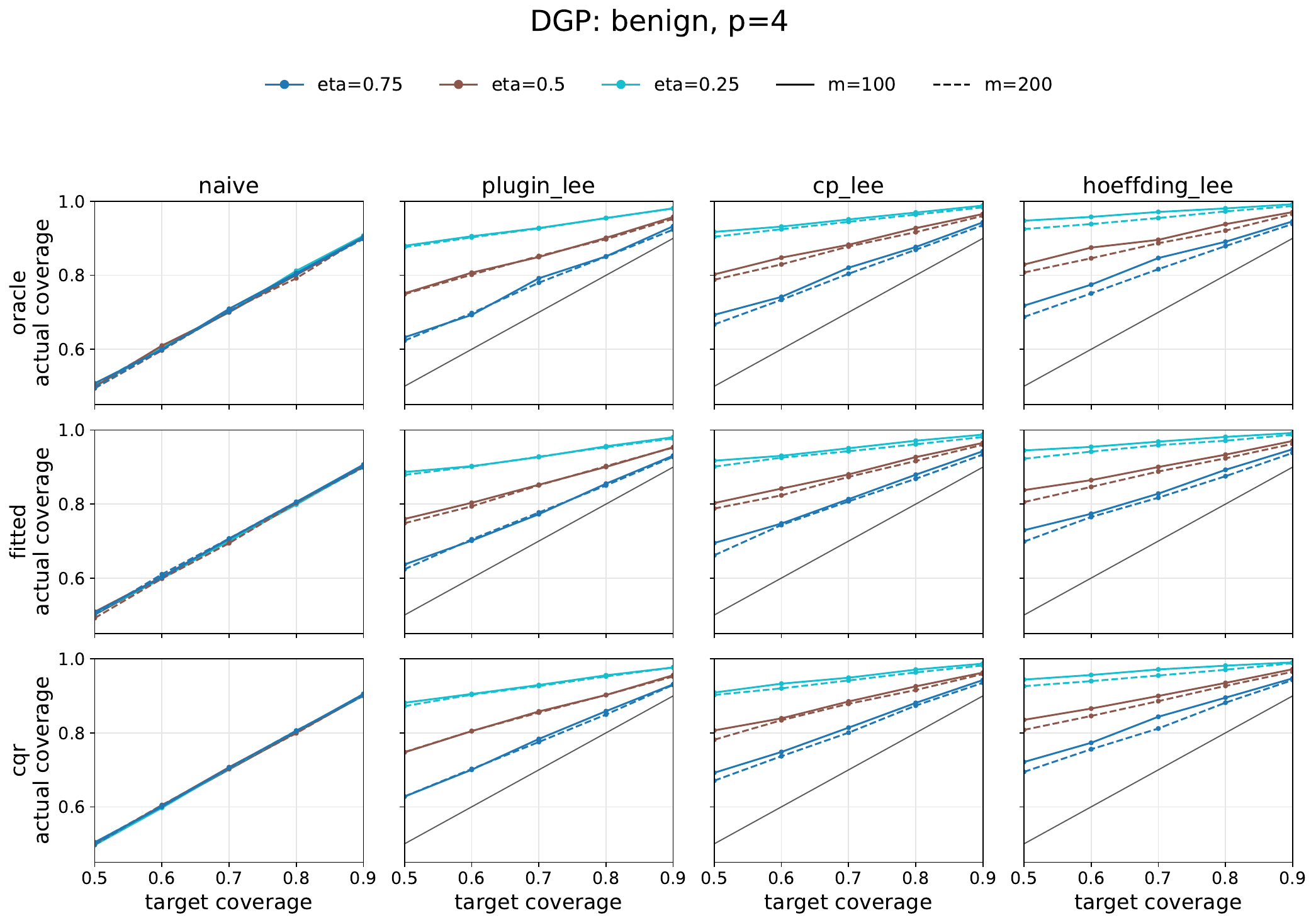}}
\begin{minipage}{0.90\textwidth}
\footnotesize
 \textit{Notes:} Columns report naive conformal prediction, plug-in Lee, CP-Lee, and Hoeffding-Lee; rows report oracle residual, fitted residual, and CQR scores. Naive coverage is close to the diagonal because $P=Q$. The adjusted procedures over-cover because they protect against shifts not realized in this design.
\end{minipage}
\label{fig:dgp-benign}
\end{figure}

\begin{figure}[p]
\centering
\caption{Coverage under conditional-tail selection}
\includegraphics[width=1 \linewidth]{\detokenize{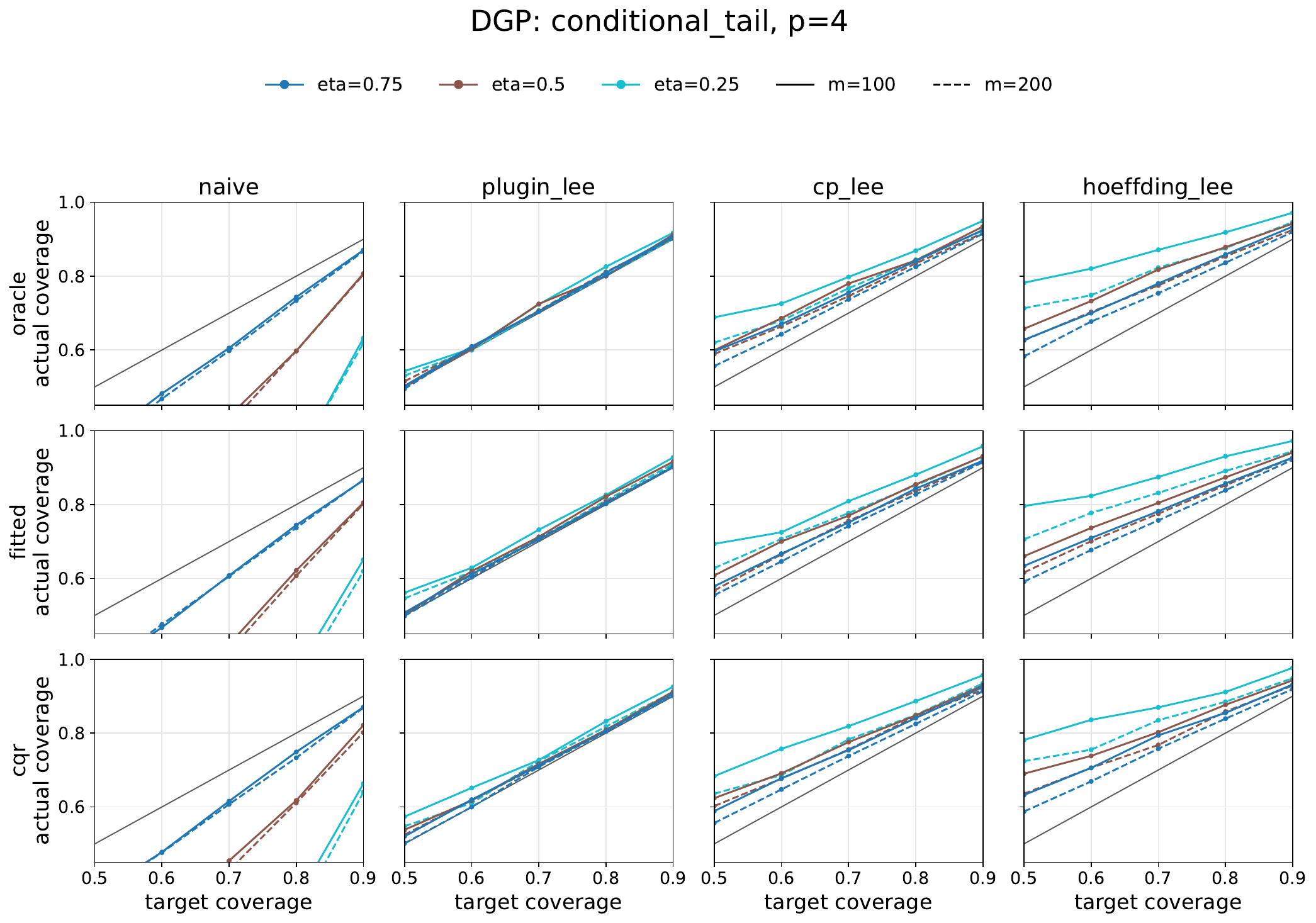}}
\begin{minipage}{0.90\textwidth}
\footnotesize
 \textit{Notes:} Naive conformal prediction under-covers increasingly as $\pi$ falls. Plug-in Lee remains close to the target, whereas CP-Lee and Hoeffding-Lee are conservative.
\end{minipage}
\label{fig:dgp-conditional_tail}
\end{figure}

\begin{figure}[p]
\centering
\caption{Coverage under unconditional-tail selection}
\includegraphics[width=1 \linewidth]{\detokenize{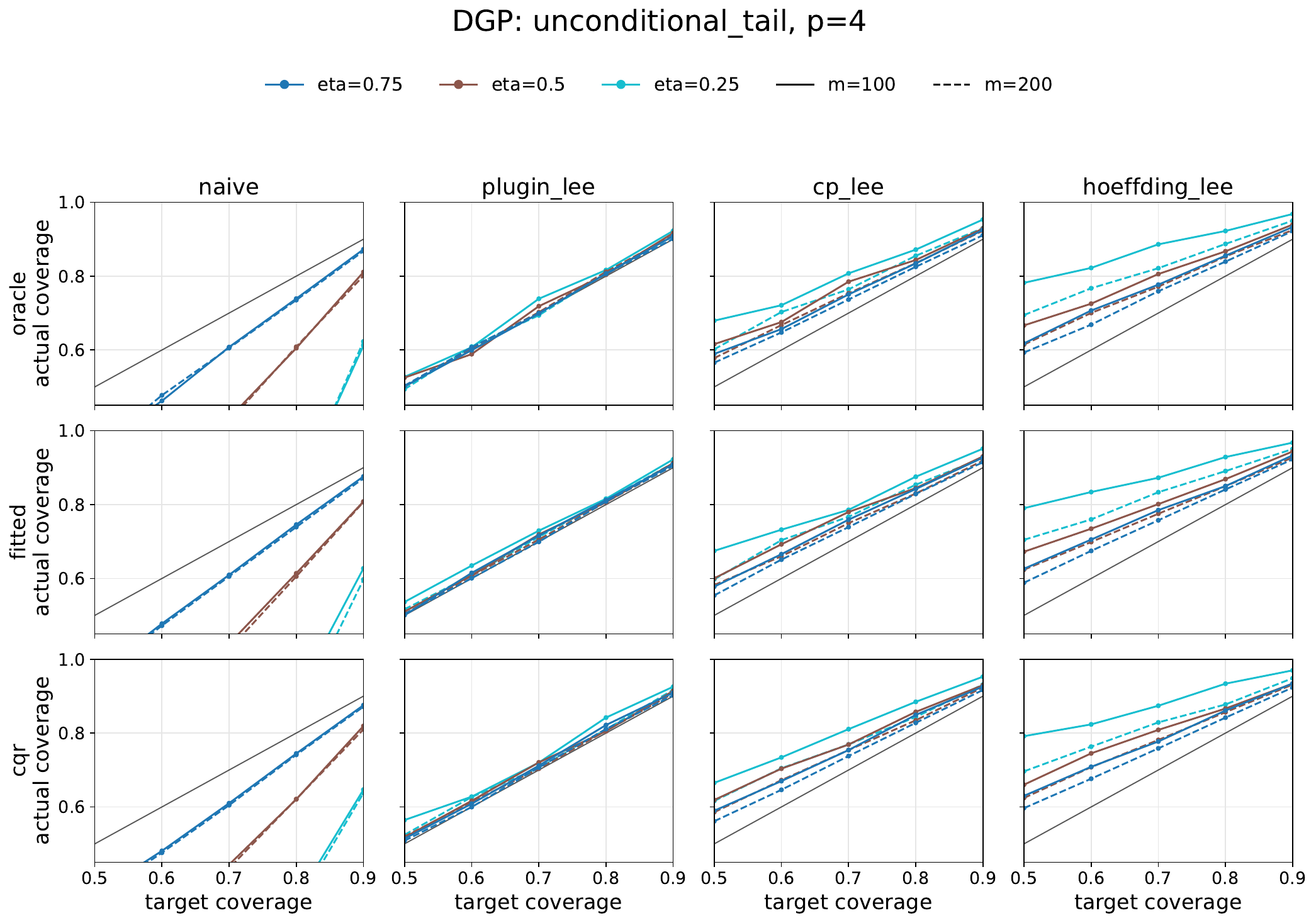}}
\begin{minipage}{0.90\textwidth}
\footnotesize
 \textit{Notes:} The design changes both the residual tail and the covariate distribution. The adjusted procedures remain valid or conservative, whereas naive conformal prediction under-covers.
\end{minipage}
\label{fig:dgp-unconditional_tail}
\end{figure}

\begin{figure}[p]
\centering
\caption{Coverage under smooth heterogeneous selection}
\includegraphics[width=1 \linewidth]{\detokenize{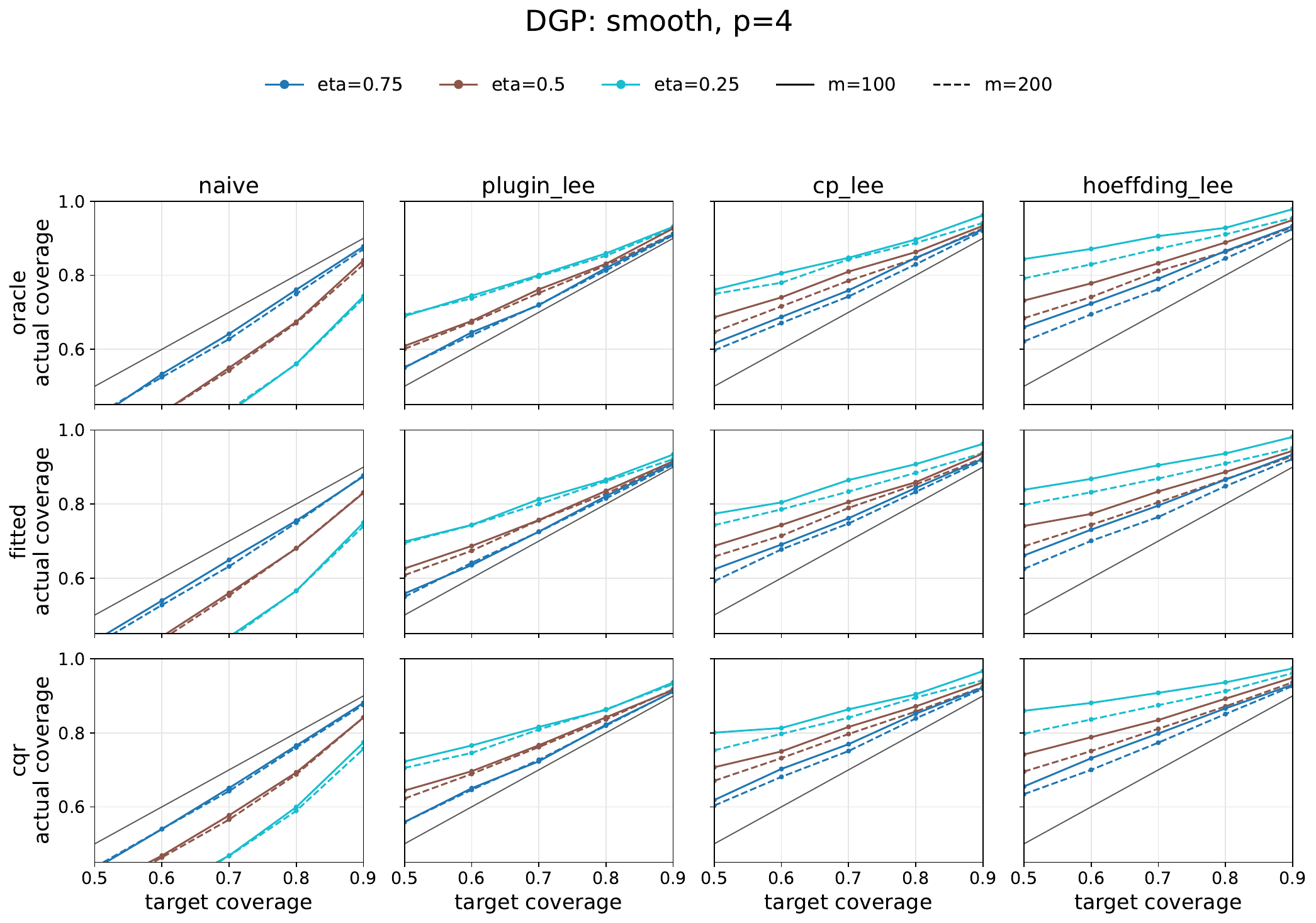}}
\begin{minipage}{0.90\textwidth}
\footnotesize
 \textit{Notes:} Naive conformal prediction under-covers moderately but systematically. Plug-in Lee remains close to the target, and the lower bound procedures are conservative.
\end{minipage}
\label{fig:dgp-smooth}
\end{figure}


\clearpage

\section{Empirical Illustration: The National Job Corps Study}
\label{sec:empirical}

This section examines how the Lee adjustment changes counterfactual prediction sets in the National Job Corps Study. We first describe wage observation and the assignment design. We then compare prediction sets across calibration rules and assess their sensitivity to the score, the sample split, and the outcome definition. 

\subsection{Data, selection, and the predictive target}
\label{sec:empirical_data}

The application uses the public data from the National Job Corps Study considered by \citet{lee2009training}. Treatment denotes assignment to the program group. Thus, the corresponding potential outcome contrast concerns assignment to program access. The analysis begins with the roster of 15,386 study records and excludes four duplicate randomizations. The resulting sample contains 5,977 controls and 9,405 treated units. 

The primary period of interest is week 130 after randomization. Let $H_i$ and $E_i$ denote recorded hours and earnings in that week. We define
\[
 S_i=\mathbf 1\{H_i>0,\ E_i>0\},
 \qquad
 Y_i=\log(E_i/H_i)\quad\text{when }S_i=1.
\]
A record without valid positive hours and earnings remains in the denominator with $S_i=0$. Therefore, selection combines positive employment earnings with the availability of usable wage information. Under monotone selection, the relevant principal stratum consists of units whose wages would be observed under either assignment. It need not coincide with the stratum of units who would be employed under either assignment.

\begin{table}[t]
\centering
\caption{Sample construction and observed wages at week 130}
\label{tab:empirical_sample}
\begin{threeparttable}
\small
\begin{tabular}{lrr}
\toprule
 & Control & Treated \\
\midrule
Roster after duplicate exclusion & 5,977 & 9,405 \\
Baseline record available & 5,514 & 8,809 \\
Employment timeline available & 5,416 & 8,776 \\
Valid positive wage at week 130 & 2,522 & 4,253 \\
Observed selection rate & 0.4220 & 0.4522 \\
\addlinespace
Mean log hourly wage & 1.8807 & 1.9204 \\
Standard deviation of log hourly wage & 0.3812 & 0.3481 \\
Median hourly wage & 6.5000 & 6.5658 \\
Maximum hourly wage & 65.0000 & 200.0000 \\
\bottomrule
\end{tabular}
\begin{tablenotes}[flushleft]
\footnotesize
\item \textit{Notes:} Counts and selection rates use the full roster after excluding duplicates. Wage summaries use only observations with positive recorded hours and earnings in week 130. 
\end{tablenotes}
\end{threeparttable}
\end{table}

Table \ref{tab:empirical_sample} reports the sample counts and observed wages. A valid positive wage is available for 2,522 controls and 4,253 treated units. The corresponding selection rates are 0.4220 and 0.4522, a difference of 3.03 percentage points. Their ratio is 0.9331. Under the sampling and monotonicity assumptions of the paper, this ratio estimates the share of always selected units among treated units with observed wages. 
The observed mean log wage is 1.8807 among selected controls and 1.9204 among treated units with observed wages. The difference of 0.0397 is descriptive because it combines differences in assignment composition, wage observation, and potential wages. Hence, it and does not identify an average treatment effect. Median hourly wages are similar across the two groups, whereas the maximum observed wage is much larger in the treated group. The primary analysis retains these observations and uses log hourly wages.

Figure \ref{fig:empirical_selection} shows that the direction of the observed selection difference changes over time. The unweighted difference is negative in every week from 1 through 52, reaches $-0.1094$ in week 15, and is nonnegative in every week from 131 through 208. 

\begin{figure}[p]
\centering
\caption{Observed selection over time}
\includegraphics[width=0.91\linewidth]{ 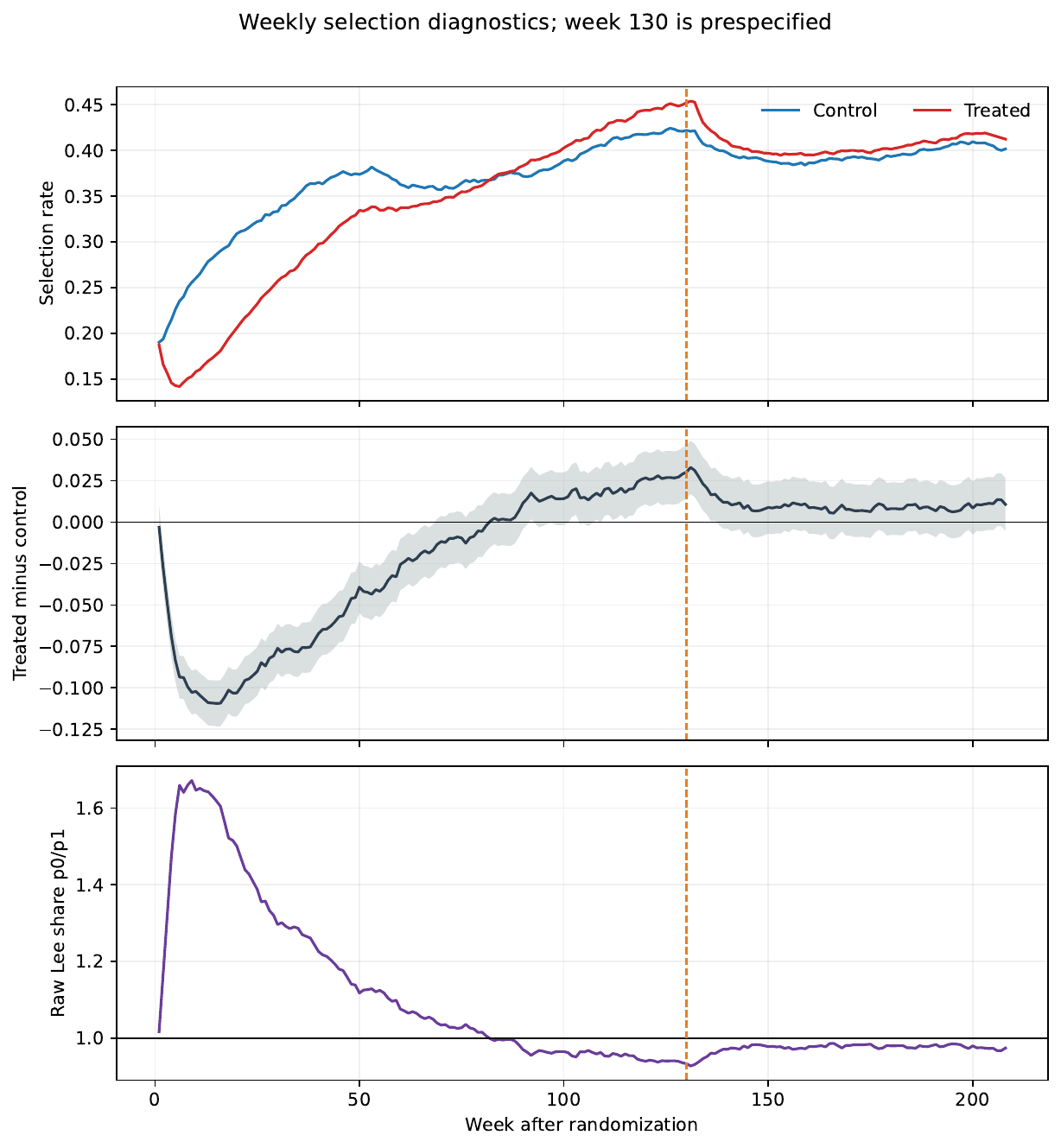}
\begin{minipage}{0.90\textwidth}
\footnotesize
 \textit{Notes:} The panels report unweighted selection rates, their difference, and their raw ratio. The vertical line marks week 130. The shaded region is the pointwise 95\% Wald interval computed from unweighted status counts. 
\end{minipage}
\label{fig:empirical_selection}
\end{figure}

\subsection{Prediction rules and calibration}
\label{sec:empirical_implementation}
The fixed split uses 2,126 treated units with observed wages for training and 2,127 for calibration. The target records are the 2,522 selected controls. Both score classes use the same split and 42 baseline or randomization features. 

The first score is the absolute residual from ridge regression, with penalty 10. The second is the conformalized quantile regression score of \citet{romano2019conformalized}. Its lower and upper prediction functions estimate the 0.05 and 0.95 quantiles using gradient boosting. Each fit uses 180 trees, learning rate 0.035, maximum depth two, and at least 20 observations per terminal node. The two fitted quantiles are ordered before scoring. These tuning choices are fixed in the implementation, and the fitted prediction functions are held constant across calibration methods within each score class.

Let $\alpha_c$ denote the miscoverage parameter used in the conformal step and let $\pi_{\rm used}$ denote its calibration share. The implemented index and threshold are
\[
 k=\left\lceil(m+1)(1-\alpha_c\pi_{\rm used})\right\rceil,
 \qquad
 \widehat t=V_{(k)},
 \qquad m=2{,}127.
\]
Naive conformal prediction sets $\pi_{\rm used}=1$, whereas the plug-in rule uses 0.9331. The remaining rules use the Clopper-Pearson or Hoeffding construction from the appendix, denoted CP and Hoeffding, respectively. Each confidence construction divides its error budget $\delta_\pi$ equally between the two assignment status.

\begin{table}[t]
\centering
\caption{Calibration inputs and score thresholds at week 130}
\label{tab:empirical_calibration}
\begin{threeparttable}
\small
\setlength{\tabcolsep}{4pt}
\begin{tabular}{lrrrrrrr}
\toprule
Rule & $\alpha_c$ & $\delta_\pi$ & $\pi_{\rm used}$ & $1-\alpha_c\pi_{\rm used}$ & $k$ & $\widehat t_R$ & $\widehat t_C$ \\
\midrule
Naive & 0.10 & \text{n.a.} & 1.0000 & 0.9000 & 1,916 & 0.5114 & 0.0585 \\
Plug-in & 0.10 & \text{n.a.} & 0.9331 & 0.9067 & 1,930 & 0.5248 & 0.0709 \\
CP & 0.10 & 0.05 & 0.8855 & 0.9115 & 1,940 & 0.5477 & 0.0889 \\
Hoeffding & 0.10 & 0.05 & 0.8674 & 0.9133 & 1,944 & 0.5522 & 0.0968 \\
CP & 0.09 & 0.01 & 0.8711 & 0.9216 & 1,962 & 0.5887 & 0.1312 \\
Hoeffding & 0.09 & 0.01 & 0.8548 & 0.9231 & 1,965 & 0.5923 & 0.1347 \\
\bottomrule
\end{tabular}
\begin{tablenotes}[flushleft]
\footnotesize
\item  \textit{Notes:} All rules use 2,127 calibration observations. Subscripts $R$ and $C$ denote the ridge residual and conformalized quantile regression scores.  The first CP and Hoeffding pair uses $(\alpha_c,\delta_\pi)=(0.10,0.05)$. The second pair uses $(0.09,0.01)$. These pairs give unconditional lower bounds 0.85 and 0.90, respectively. 
\end{tablenotes}
\end{threeparttable}
\end{table}

Table \ref{tab:empirical_calibration} distinguishes the two error allocations. Using $\alpha_c=0.10$ and $\delta_\pi=0.05$ gives an unconditional lower bound of $1-\alpha_c-\delta_\pi=0.85$ under the assumptions of the paper. The alternative allocation uses $\alpha_c=0.09$ and $\delta_\pi=0.01$, which gives the stated bound 0.90 under those assumptions. In the tables below, the pair following CP or Hoeffding records $(\alpha_c,\delta_\pi)$. 

The plug-in adjustment moves the target score quantile from 0.9000 to 0.9067 and the order statistic from 1,916 to 1,930. The CP rule with allocation $(0.09,0.01)$ uses the smaller share 0.8711 and the larger quantile 0.9216, giving rank 1,962. Figure \ref{fig:empirical_scores} shows where these cutoffs fall in each score distribution. 

\begin{figure}[t]
\centering
\caption{Calibration scores and selected thresholds}
\includegraphics[width=\linewidth]{ 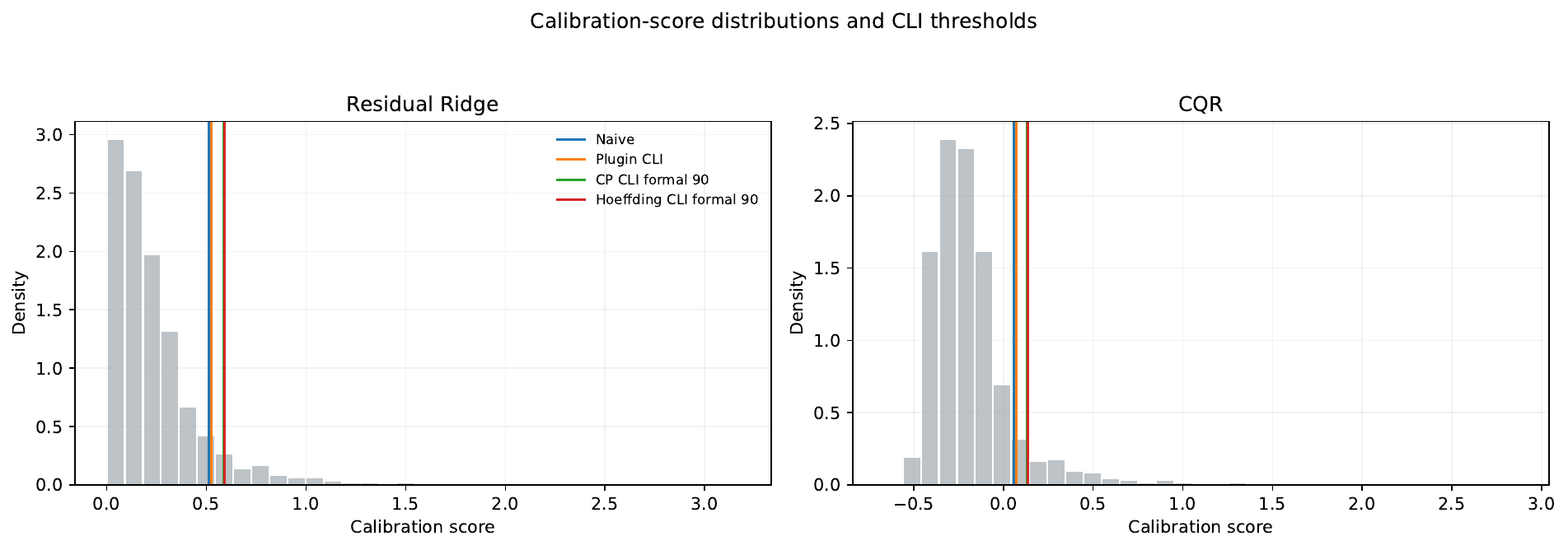}
\begin{minipage}{0.90\textwidth}
\footnotesize
 \textit{Notes:} The panels show the 2,127 calibration scores for ridge regression and conformalized quantile regression. Vertical lines identify the naive, plug-in, CP, and Hoeffding thresholds shown in the archived figure. The last two use $(\alpha_c,\delta_\pi)=(0.09,0.01)$. Their labels refer to an error allocation under the theorem's assumptions, not established empirical coverage.
\end{minipage}
\label{fig:empirical_scores}
\end{figure}

\subsection{Covariates and construction of individual prediction sets}
\label{sec:empirical_covariates}

For the primary wage analysis, $X_i$ contains 42 source variables describing baseline characteristics and assignment timing. Table \ref{tab:empirical_covariates} lists these variables, of which 16 are processed as numeric and 26 as categorical. The same covariates enter the ridge and quantile regression models. They describe demographics, education, household background, employment and welfare histories, health, criminal history, area, and the week of randomization. The baseline hourly wage variable \texttt{HRWAGER} is included as a predictor. 

\begin{table}[t]
\centering
\caption{Covariates used in the primary wage prediction models}
\label{tab:empirical_covariates}
\begin{threeparttable}
\small
\begin{tabular}{@{}p{0.25\linewidth}p{0.71\linewidth}@{}}
\toprule
Component & Source variables \\
\midrule
\raggedright Demographics & \raggedright \texttt{AGE}$^{\mathrm n}$, \texttt{FEMALE}, \texttt{RACE\_ETH}, \texttt{NTV\_LANG} \tabularnewline
\addlinespace[3pt]
\raggedright Education & \raggedright \texttt{HGC}$^{\mathrm n}$, \texttt{HGC\_MOTH}$^{\mathrm n}$, \texttt{HGC\_FATH}$^{\mathrm n}$, \texttt{HS\_D}, \texttt{GED\_D}, \texttt{VOC\_D}, \texttt{OTH\_DEG}, \texttt{INSCHOOL} \tabularnewline
\addlinespace[3pt]
\raggedright Household and childhood background & \raggedright \texttt{NCHLD}$^{\mathrm n}$, \texttt{NUMB\_HH}$^{\mathrm n}$, \texttt{HH14}, \texttt{WELF\_KID}, \texttt{M\_WORK14}, \texttt{F\_WORK14}, \texttt{MARRIAGE}, \texttt{HASCHLD}, \texttt{R\_HEAD}, \texttt{HOUS\_ARR}, \texttt{PAY\_RENT} \tabularnewline
\addlinespace[3pt]
\raggedright Employment, wages, and earnings at baseline & \raggedright \texttt{NUMBJOBS}$^{\mathrm n}$, \texttt{YR\_WORK1}$^{\mathrm n}$, \texttt{EARN\_YR}$^{\mathrm n}$, \texttt{MOSINJOB}$^{\mathrm n}$, \texttt{HRSWK\_JR}$^{\mathrm n}$, \texttt{HRWAGER}$^{\mathrm n}$, \texttt{WKEARNR}$^{\mathrm n}$, \texttt{EVWORKB}, \texttt{REC\_JOB}, \texttt{OCC\_R} \tabularnewline
\addlinespace[3pt]
\raggedright Welfare history & \raggedright \texttt{MOS\_ANYW}$^{\mathrm n}$, \texttt{GOT\_ANYW} \tabularnewline
\addlinespace[3pt]
\raggedright Health & \raggedright \texttt{HEALTH}, \texttt{SICK} \tabularnewline
\addlinespace[3pt]
\raggedright Criminal history & \raggedright \texttt{EVARRST1}, \texttt{N\_GUILTY}$^{\mathrm n}$ \tabularnewline
\addlinespace[3pt]
\raggedright Area and assignment timing & \raggedright \texttt{TYPEAREA}, \texttt{JCMSA}, \texttt{RAND\_WK}$^{\mathrm n}$ \tabularnewline
\bottomrule
\end{tabular}
\begin{tablenotes}[flushleft]
\footnotesize
\item Notes: A superscript $\mathrm n$ denotes a numeric input. The remaining inputs are categorical. Variable names follow the archived implementation. \texttt{FEMALE} and \texttt{TYPEAREA} are taken from \texttt{key\_vars.sas7bdat}, and \texttt{RAND\_WK} is taken from \texttt{rand\_dat.sas7bdat}. The other 39 variables are taken from \texttt{baseline.sas7bdat}. 
\end{tablenotes}
\end{threeparttable}
\end{table}

Preprocessing is estimated using only the treated training sample. Missing numeric entries are replaced by training sample medians. Categorical entries that are not approximately integers are treated as missing, and missing categories are replaced by their most frequent training category. Missingness indicators are added for inputs with missing entries in the training sample. Numeric inputs and their missingness indicators are standardized. Categorical inputs and their missingness indicators are represented by indicator columns. Then the fitted transformation is applied to calibration observations and target controls without refitting. Write the resulting model input as $\widetilde X_i=\widehat T_{\rm train}(X_i)$. 

For the ridge specification, the fitted treated wage predictor is
\[
 \widehat\mu_1(x)
 =\widehat\beta_0+\widehat\beta^{\top}\widehat T_{\rm train}(x).
\]
Calibration uses the absolute residuals from this predictor. Given the threshold $\widehat t_R$ for a calibration rule in Table \ref{tab:empirical_calibration}, the endpoints for target individual $i$ are
\[
 \widehat L_1(X_i)=\widehat\mu_1(X_i)-\widehat t_R,
 \qquad
 \widehat U_1(X_i)=\widehat\mu_1(X_i)+\widehat t_R.
\]
Thus, the prediction rule is evaluated at that individual's covariates.

For the quantile regression specification, let $\widehat q_{1,L}(x)$ and $\widehat q_{1,U}(x)$ denote the ordered fitted functions described in Section \ref{sec:empirical_implementation}. Unlike the ridge intervals, the quantile regression intervals have widths that vary with the covariates. Their underlying quantile levels remain 0.05 and 0.95 for every calibration rule, including those with $\alpha_c=0.09$. The calibration score for an observation $(X_j,Y_j)$ is
\[
 V_{C,j}
 =\max\bigl\{\widehat q_{1,L}(X_j)-Y_j,\,
             Y_j-\widehat q_{1,U}(X_j)\bigr\}.
\]
Using the corresponding threshold $\widehat t_C$ gives
\[
 \widehat L_1(X_i)=\widehat q_{1,L}(X_i)-\widehat t_C,
 \qquad
 \widehat U_1(X_i)=\widehat q_{1,U}(X_i)+\widehat t_C.
\]
Both specifications are fitted on selected treated observations, whose membership in $\mathcal A$ is not observed. 

\begin{samepage}
For either score, let $\widehat C_1(X_i)=[\widehat L_1(X_i),\widehat U_1(X_i)]$ denote the counterfactual interval. Each target individual satisfies $D_i=0$ and $S_i=1$, so $Y_i(0)=Y_i$ is observed. Under monotone selection, this individual belongs to $\mathcal A$. Therefore, the prediction set for $\tau_i=Y_i(1)-Y_i(0)$ is 
\[
 \widehat C_{\tau,i}
 =\widehat C_\tau\bigl(X_i,Y_i(0)\bigr)
 =\bigl[\widehat L_1(X_i)-Y_i(0),\,
        \widehat U_1(X_i)-Y_i(0)\bigr].
\]
\end{samepage}
The observed control wage enters through this translation. Subtracting it shifts both endpoints by the same amount and leaves interval width unchanged.

Within a calibration rule, every ridge interval has width $2\widehat t_R$. Covariates affect its location through $\widehat\mu_1(X_i)$, and the observed control wage determines the final center $\widehat\mu_1(X_i)-Y_i(0)$. Quantile regression intervals instead have width $\widehat q_{1,U}(X_i)-\widehat q_{1,L}(X_i)+2\widehat t_C$, so their locations and widths can both vary with $X_i$. Note that although covariates enter the prediction functions, they do not yield a different always-selected share or calibration threshold for each target individual. Within a score class and calibration rule, $\pi_{\rm used}$ and the threshold are common across all selected controls. 

\subsection{Individual prediction set results}
\label{sec:empirical_results}

Table \ref{tab:empirical_sets} summarizes the individual effect sets constructed in Section \ref{sec:empirical_covariates}. Within each score class, the fitted functions, the target covariates, and the observed control wages are held fixed across calibration rules. Therefore, differences across rules arise only from the calibration thresholds.

\begin{table}[t]
\centering
\caption{Prediction set summaries for the 2,522 selected controls}
\label{tab:empirical_sets}
\begin{threeparttable}
\small
\setlength{\tabcolsep}{4pt}
\begin{tabular}{lrrrrrr}
\toprule
 & \multicolumn{3}{c}{Log wage effect set} & \multicolumn{3}{c}{Percentage of sets} \\
\cmidrule(lr){2-4}\cmidrule(lr){5-7}
Rule & \shortstack{Median\\width} & \shortstack{Median\\lower} & \shortstack{Median\\upper} & \shortstack{Above\\zero} & \shortstack{Below\\zero} & \shortstack{Include\\zero} \\
\midrule
\multicolumn{7}{l}{\textit{A. Ridge residual score}} \\[2pt]
Naive & 1.0229 & -0.4529 & 0.5700 & 5.55 & 4.28 & 90.17 \\
plug-in & 1.0495 & -0.4663 & 0.5833 & 5.35 & 4.16 & 90.48 \\
CP $(0.10,0.05)$ & 1.0954 & -0.4892 & 0.6062 & 5.04 & 4.00 & 90.96 \\
Hoeffding $(0.10,0.05)$ & 1.1043 & -0.4936 & 0.6107 & 5.04 & 3.93 & 91.04 \\
CP $(0.09,0.01)$ & 1.1773 & -0.5301 & 0.6472 & 4.60 & 3.21 & 92.19 \\
Hoeffding $(0.09,0.01)$ & 1.1846 & -0.5338 & 0.6508 & 4.52 & 3.13 & 92.35 \\
\addlinespace[5pt]
\multicolumn{7}{l}{\textit{B. Conformalized quantile regression score}} \\[2pt]
Naive & 0.9799 & -0.3651 & 0.6203 & 6.07 & 4.20 & 89.73 \\
plug-in & 1.0047 & -0.3775 & 0.6327 & 5.79 & 3.85 & 90.36 \\
CP $(0.10,0.05)$ & 1.0407 & -0.3955 & 0.6507 & 5.67 & 3.57 & 90.76 \\
Hoeffding $(0.10,0.05)$ & 1.0565 & -0.4034 & 0.6586 & 5.63 & 3.49 & 90.88 \\
CP $(0.09,0.01)$ & 1.1253 & -0.4378 & 0.6930 & 5.15 & 2.74 & 92.11 \\
Hoeffding $(0.09,0.01)$ & 1.1322 & -0.4412 & 0.6965 & 5.11 & 2.70 & 92.19 \\
\bottomrule
\end{tabular}
\begin{tablenotes}[flushleft]
\footnotesize
\item Notes: The pair following CP or Hoeffding is $(\alpha_c,\delta_\pi)$. Above and below zero require strict exclusion of zero. \end{tablenotes}
\end{threeparttable}
\end{table}

Panel A of Table \ref{tab:empirical_sets} reports the ridge results. The plug-in interval has width 1.0495 log points, compared with 1.0229 for naive conformal prediction. The increase is 2.61\%. The CP rule with allocation $(0.09,0.01)$ increases the width to 1.1773, or 15.10\% above the naive width. Thus, the point estimate of the rate ratio produces a modest adjustment in this split, whereas accounting for uncertainty in that ratio and allocating the error budget produces a larger change.

Furthermore, Panel A reports that most intervals contain zero. Under the plug-in ridge rule, 5.35\% lie wholly above zero, 4.16\% lie wholly below zero, and 90.48\% contain zero. The CP rule with allocation $(0.09,0.01)$ raises the last fraction to 92.19\%. Figure \ref{fig:empirical_endpoints} displays the corresponding endpoint distributions. These results show that the reported sets generally do not strictly distinguish the sign of an individual effect. Note that they do not estimate the proportion of positive or negative ITEs. Also, the fraction of intervals containing zero is not a coverage rate although it is close to 90.

Panel B shows that the qualitative comparison is unchanged with the quantile regression score. Its plug-in median width is 1.0047, compared with 1.0495 for the ridge score, and 90.36\% of its sets contain zero. The corresponding CP median width is 1.1253 and 92.11\% of its sets contain zero. 

\begin{figure}[t]
\centering
\caption{Distributions of the ridge interval endpoints}
\includegraphics[width=\linewidth]{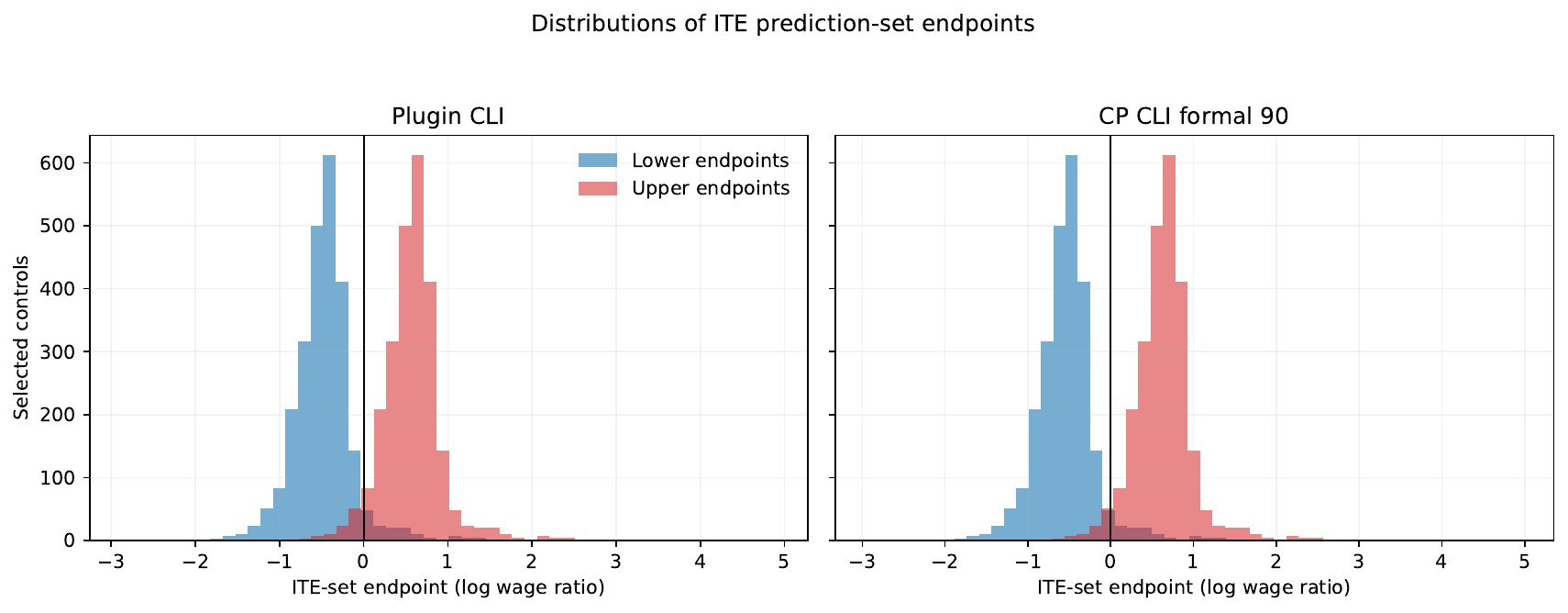}
\begin{minipage}{0.90\textwidth}
\footnotesize
 \textit{Notes:} Each panel summarizes lower and upper endpoints across the selected controls. The left panel uses the plug-in share and the right panel uses CP with $(\alpha_c,\delta_\pi)=(0.09,0.01)$. The histograms describe prediction set endpoints.
\end{minipage}
\label{fig:empirical_endpoints}
\end{figure}

\subsection{Observable diagnostics and sensitivity}
\label{sec:empirical_diagnostics}
Actual counterfactual coverage cannot be observed because treated potential wages are missing for selected control records. Therefore, we examine prediction of baseline outcomes that are observed in both assignment groups. The principal diagnostic uses log baseline hourly wages among records with a positive baseline wage and a valid wage at week 130. The ridge model is trained on eligible treated records and evaluated on eligible control records. Baseline hourly wages and the related earnings, hours, and job duration predictors specified in the code are excluded from this fit to avoid direct or mechanical prediction of the outcome from itself.

\begin{table}[t]
\centering
\caption{Prediction of baseline log hourly wages across 100 splits}
\label{tab:empirical_placebo}
\begin{threeparttable}
\small
\begin{tabular}{lrrr}
\toprule
Rule & Mean coverage & Standard deviation & Mean length \\
\midrule
Naive & 0.8914 & 0.0090 & 0.7871 \\
plug-in & 0.8981 & 0.0090 & 0.8175 \\
CP $(0.10,0.05)$ & 0.9041 & 0.0094 & 0.8462 \\
Hoeffding $(0.10,0.05)$ & 0.9062 & 0.0092 & 0.8552 \\
CP $(0.09,0.01)$ & 0.9165 & 0.0079 & 0.9091 \\
Hoeffding $(0.09,0.01)$ & 0.9185 & 0.0078 & 0.9209 \\
\bottomrule
\end{tabular}
\begin{tablenotes}[flushleft]
\footnotesize
\item  \textit{Notes:} Each split uses 1,373 training records, 1,374 calibration records, and the same 1,626 eligible selected controls. The outcome is an observed baseline log wage. Standard deviations measure variation across partitions of the same data. Each method retains the calibration share computed from the full week 130 roster.
\end{tablenotes}
\end{threeparttable}
\end{table}

Table \ref{tab:empirical_placebo} reports average observed coverage across 100 splits. Naive conformal prediction covers 0.8914 of eligible controlled outcomes on average, compared with 0.8981 for the plug-in rule and 0.9165 for CP with allocation $(0.09,0.01)$. Mean interval lengths are 0.7871, 0.8175, and 0.9091, respectively. Within a split, the fitted model is common across methods and the Lee adjustments use weakly larger thresholds. Thus, the increase in observed coverage reflects nested prediction sets. 

Figure \ref{fig:empirical_placebo} reports the corresponding comparisons for baseline annual earnings, age, schooling grade, and prior convictions. The plug-in mean coverages are 0.9023, 0.9155, 0.9206, and 0.9053, respectively. The annual earnings outcome is $\log(1+\text{earnings})$, whereas the last three outcomes remain in their recorded units. The implementation codes an unrecorded count as zero only when the baseline record is available. These exercises restrict the sample according to baseline outcome availability while retaining the calibration shares from the full roster. 

Table \ref{tab:empirical_placebo} also assesses partition sensitivity and the result is stable in their ordering across alternative splits. Across 100 splits of the week 130 treated sample, mean ridge widths are 0.9849 for naive conformal prediction, 1.0204 for the plug-in rule, and 1.1222 for CP with allocation $(0.09,0.01)$. The plug-in width ranges from 0.9573 to 1.1032, which contains the fixed split value of 1.0495. 

\begin{table}[t]
\centering
\caption{Selection rates and ridge interval widths at alternative endpoints}
\label{tab:empirical_time}
\begin{threeparttable}
\small
\setlength{\tabcolsep}{5pt}
\begin{tabular}{rrrrrrr}
\toprule
 & \multicolumn{3}{c}{Unweighted selection} & \multicolumn{3}{c}{Median interval width} \\
\cmidrule(lr){2-4}\cmidrule(lr){5-7}
Week & $\widehat p_0$ & $\widehat p_1$ & $\widehat\pi$ & Naive & plug-in & CP $(0.09,0.01)$ \\
\midrule
104 & 0.3960 & 0.4106 & 0.9644 & 1.0456 & 1.0576 & 1.1735 \\
130 & 0.4220 & 0.4522 & 0.9331 & 1.0229 & 1.0495 & 1.1773 \\
156 & 0.3843 & 0.3960 & 0.9706 & 1.1076 & 1.1290 & 1.2556 \\
182 & 0.3950 & 0.4030 & 0.9802 & 1.1594 & 1.1771 & 1.2692 \\
208 & 0.4017 & 0.4124 & 0.9740 & 1.0832 & 1.0966 & 1.1973 \\
\bottomrule
\end{tabular}
\begin{tablenotes}[flushleft]
\footnotesize
\item  \textit{Notes:} The full roster supplies the denominator at each endpoint. The source sample, target sample, prediction fit, and calibration fold are reconstructed for that endpoint. Widths are in log hourly wage units. Each endpoint defines its own wage observation event and hence its own candidate always observed stratum.
\end{tablenotes}
\end{threeparttable}
\end{table}

Table \ref{tab:empirical_time} reports the other endpoints fixed in the analysis specification. Since wage observation changes over time, these comparisons do not track a fixed principal stratum. Instead, they assess the numerical behavior of the same construction under different endpoint definitions. The unweighted selection difference is positive at weeks 104, 156, 182, and 208, and the corresponding rate ratios range from 0.9644 to 0.9802. The plug-in ridge widths range from 1.0576 to 1.1771 at these endpoints. 

A smoothed outcome uses the median log hourly wage over weeks 118 through 130 and requires valid wages in at least seven of those weeks. Its selection rates are 0.4204 and 0.4503, rate ratio is 0.9337, and plug-in ridge width is 1.0431. Restricting recorded hours at week 130 to at most 168 removes one selected record from each status and gives a plug-in width of 1.1082. Both exercises reconstruct and refit the prediction procedure. Hence, their width differences combine changes in the observations with changes in the split and fitted rule. 


\begin{figure}[p]
\centering
\caption{Observed baseline prediction coverage across 100 splits}
\includegraphics[width=\linewidth,height=0.76\textheight,keepaspectratio]{ 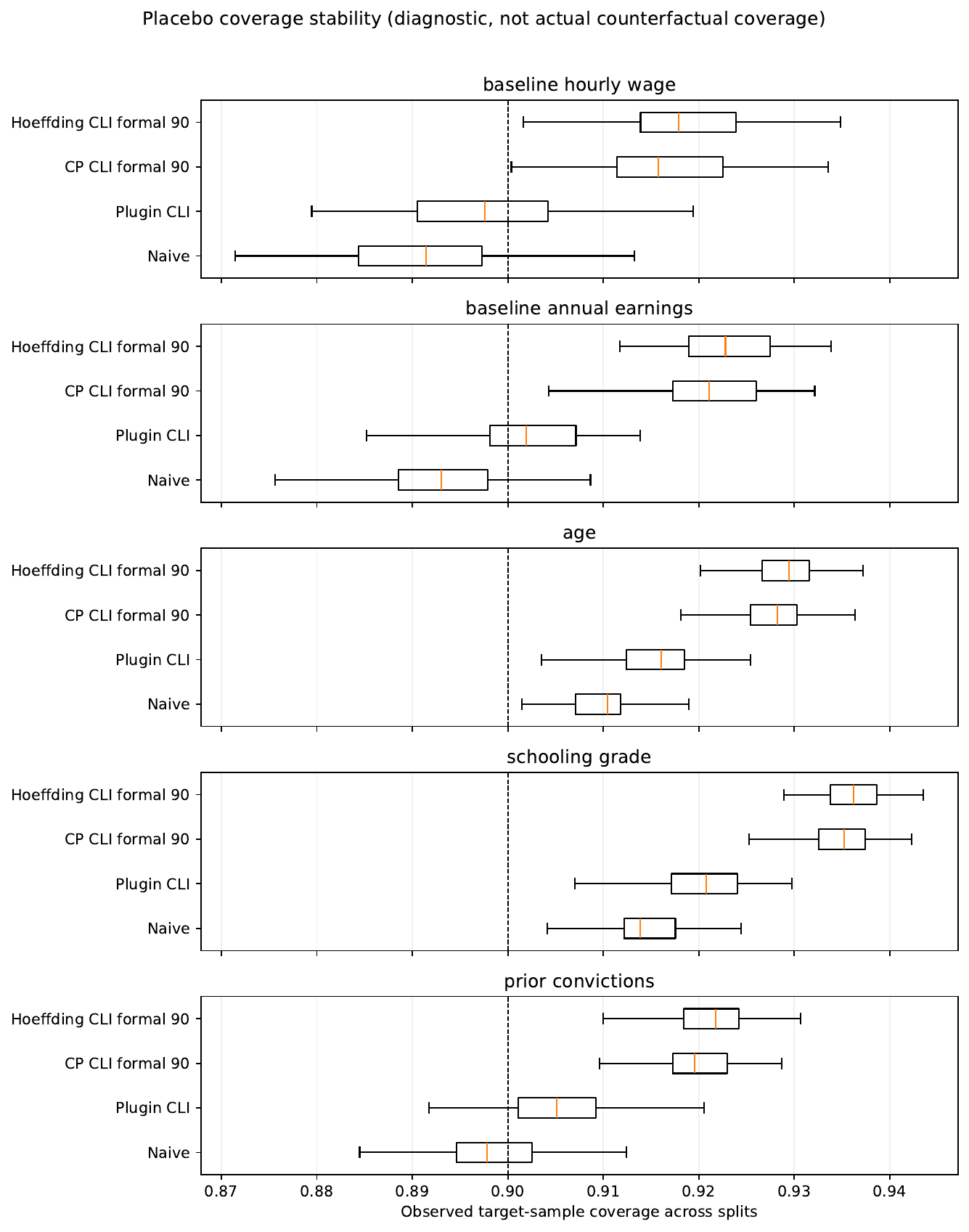}
\begin{minipage}{0.90\textwidth}
\footnotesize
 \textit{Notes:} The panels reproduce the five placebo exercises in the archived analysis. The vertical line marks 0.90 as a reference value. Boxplots summarize variation across sample partitions, with outliers omitted from the display. The plotted coverage concerns observed baseline outcomes and is not actual counterfactual or individual treatment effect coverage. The ``formal 90'' labels denote the $(0.09,0.01)$ allocation.
\end{minipage}
\label{fig:empirical_placebo}
\end{figure}

\clearpage

\section{Conclusion}

This paper develops conformalized Lee inference for counterfactual prediction under monotone sample selection. Random assignment and the selection rates identify the always-selected share $\pi=p_0/p_1$, while the latent identity of always-selected treated units remains unknown. The corresponding target distribution satisfies $Q_0\ll P$ and $dQ_0/dP\leq1/\pi$. Lee ambiguity set based on the likelihood ratio restriction is sharp relative to the reduced-information $(P,p_0,p_1)$.

The likelihood ratio bound changes the ordinary split-conformal cutoff from the $(1-\alpha)$ score quantile to the $(1-\alpha\pi)$ quantile. The resulting set has finite-sample marginal coverage over the Lee ambiguity set, and translation by an observed control outcome gives marginal ITE coverage for an independent selected-control draw. The same cutoff is minimax optimal over the reduced-information region. The simulations show that naive conformal prediction can under-cover under selection-induced shifts, whereas the adjusted procedures restore coverage. Lower confidence bounds for $\pi$ provide conservative coverage at the cost of wider and occasionally uninformative sets. 

The empirical illustration using the National Job Corps Study data illustrates how conformalized Lee inference combines covariate based wage prediction with a selection adjustment. We use ridge and quantile regression to predict treated log wages from baseline characteristics and assignment timing. For each selected control, we evaluate the calibrated treated log-wage interval at that individual's covariates and subtract the observed control log wage from both endpoints. The resulting prediction set concerns the individual wage effect of assignment to program access and requires prediction only of the missing treated outcome. Both nonconformity scores produce prediction sets that span both wage gains and wage losses for most selected controls.

Two extensions follow directly from the analysis. First, the full observed distribution identifies the always-selected covariate marginal and the conditional restriction
\[
Q_x\ll P_x,
\qquad
0\leq\frac{dQ_x}{dP_x}(y)\leq\frac{1}{\pi(x)}.
\]
A covariate-adaptive procedure could use this information to shorten prediction sets. Second, sensitivity analysis for violations of monotonicity could replace exact nesting with a controlled amount of bidirectional selection. Both extensions should retain the distinction between finite-sample marginal coverage and individualized conditional coverage.

\newpage
\section{Appendix}
\label{sec:appendix}

\subsection{Consistency of the plug-in estimator for the share of always-selected units}

Proposition \ref{prop:threshold-consistency} connects the sample cutoff with the population cutoff given that $m_n \overset{p}{\longrightarrow} \infty
$ and $\widehat\pi_n \overset{p}{\longrightarrow} \pi$. For the count-based estimator and confidence bounds, we impose the following condition.

\begin{cond}[Conditional armwise binomial sampling]
\label{cond:selection-counts}
For every $n$ and $d\in\{0,1\}$, conditional on the assignment vector $\mathbf D_n = (D_1,\ldots,D_n)$,
\[
\sum_{i=1}^n\mathbbm 1\{D_i=d\}S_i(d)
\sim
\operatorname{Binomial}\!\left(\sum_{i=1}^n\mathbbm 1\{D_i=d\},p_d\right)
\quad\text{almost surely}.
\]
\end{cond}

Condition \ref{cond:selection-counts} is the exact count distribution requirement used by the plug-in estimator and the one-sided confidence bounds. It is implied by independent and identically distributed latent sampling together with independence of the assignment vector and the potential outcome vector. 

Proposition \ref{prop:plugin-pi-consistency} establishes the required conditions for $m_n \overset{p}{\longrightarrow} \infty
$ and $\widehat\pi_n \overset{p}{\longrightarrow} \pi$.

\begin{prop}[Consistency of the $\widehat \pi$ plug-in Lee share]
\label{prop:plugin-pi-consistency}
For $d\in\{0,1\}$, define
\[
N_{d,n}=\sum_{i=1}^n\mathbbm 1\{D_i=d\},
\qquad
M_{d,n}=\sum_{i=1}^n\mathbbm 1\{D_i=d,S_i=1\}.
\]
Suppose Condition \ref{cond:selection-counts} holds and
\[
\frac{N_{d,n}}{n}\overset{p}{\longrightarrow}\rho_d\in(0,1),
\qquad d\in\{0,1\}.
\]
Define 
\[
\widehat p_{d,n}
=
\begin{cases}
M_{d,n}/N_{d,n},
&N_{d,n} > 0, \\[3pt]
0, &N_{d,n} =0,
\end{cases}
\]
and
\[
\widehat\pi_n
=
\begin{cases}
\min\{1,\widehat p_{0,n}/\widehat p_{1,n}\},
&\widehat p_{1,n}>0,\\[3pt]
0,&\widehat p_{1,n}=0.
\end{cases}
\]
Then
\[
\widehat p_{d,n}\overset{p}{\longrightarrow}p_d,
\qquad d\in\{0,1\},
\qquad\text{and}\qquad
\widehat\pi_n\overset{p}{\longrightarrow}\pi.
\]
Consequently, Corollary \ref{coro:plugin-asymptotic-coverage} applies to the
plug-in procedure.

Furthermore, if the random calibration split also satisfies
\[
\frac{m_n}{M_{1,n}}\overset{p}{\longrightarrow}\lambda\in(0,1),
\]
where $M_{1,n}=|I_1|$ is the number of treated-selected observations, then
\[
\frac{m_n}{n}\overset{p}{\longrightarrow}\lambda\rho_1p_1>0.
\]
Hence $m_n\overset{p}{\longrightarrow}\infty$.
\end{prop}

\newpage
\subsection{Practical lower bound Constructions for the share of always-takers}

Theorem \ref{thm1:coverage} preserves conditional $1-\alpha$ coverage whenever the calibration value satisfies $\widehat\pi\leq\pi$. This subsection gives two one-sided confidence bounds that can be used as $\widehat\pi$. If either treatment status is empty, both constructions set $\widehat\pi=0$, which gives $k=m+1$ and the uninformative prediction set. Hence, the formulas below assume $N_0,N_1\geq1$.

\paragraph{Clopper--Pearson construction.}
Define
\[
N_d=\sum_{i=1}^n\mathbf 1\{D_i=d\},
\qquad
M_d=\sum_{i=1}^n\mathbf 1\{D_i=d,S_i=1\},
\qquad d\in\{0,1\}.
\]

Under Condition \ref{cond:selection-counts}, conditional on the assignment vector, and therefore conditional on $N_d$, $M_d$ is binomial with size $N_d$ and success probability $p_d$. Fix $\delta_\pi\in(0,1)$ and choose $\delta_0,\delta_1>0$ such that $\delta_0+\delta_1=\delta_\pi$. A one-sided Clopper--Pearson lower confidence bound for $p_0$ is

\[
p_{0,L}^{\rm CP}
=
\begin{cases}
0, & M_0=0,\\[4pt]
B^{-1}(\delta_0;M_0,N_0-M_0+1), & M_0>0,
\end{cases}
\]
where $B^{-1}(q;a,b)$ is the $q$ quantile of a $\operatorname{Beta}(a,b)$ distribution. A one-sided upper confidence bound for $p_1$ is
\[
p_{1,U}^{\rm CP}
=
\begin{cases}
1, & M_1=N_1,\\[4pt]
B^{-1}(1-\delta_1;M_1+1,N_1-M_1), & M_1<N_1.
\end{cases}
\]
Set
\[
\widehat\pi_{\rm CP}
=
\min\left\{1,\frac{p_{0,L}^{\rm CP}}{p_{1,U}^{\rm CP}}\right\}.
\]
The simultaneous event
\[
p_{0,L}^{\rm CP}\leq p_0,
\qquad
p_1\leq p_{1,U}^{\rm CP}
\]
has probability at least $1-\delta_\pi$. On this event, $\widehat\pi_{\rm CP}\leq p_0/p_1=\pi$. Therefore, using $\widehat\pi=\widehat\pi_{\rm CP}$ in Algorithm 1 gives conditional coverage of at least $1-\alpha$ on an event with probability at least $1-\delta_\pi$. Averaging over this event gives the unconditional guarantee
\[
\Pr\{Y_{m+1}\in\widehat C_1(X_{m+1})\}
\geq
1-\alpha-\delta_\pi.
\]
To target unconditional coverage $1-\alpha$, the conformal step may instead use miscoverage $\alpha-\delta_\pi$, provided $\delta_\pi<\alpha$.

\paragraph{Hoeffding construction.}
A closed-form alternative follows from Hoeffding's inequality. For $d\in\{0,1\}$, define
\[
\widehat p_d=\frac{M_d}{N_d},
\qquad
\varepsilon_d
=
\sqrt{\frac{\log(1/\delta_d)}{2N_d}}.
\]
Set
\[
p_{0,L}^{\rm H}
=\max\{0,\widehat p_0-\varepsilon_0\},
\qquad
p_{1,U}^{\rm H}
=\min\{1,\widehat p_1+\varepsilon_1\},
\]
and
\[
\widehat\pi_{\rm H}
=
\min\left\{1,\frac{p_{0,L}^{\rm H}}{p_{1,U}^{\rm H}}\right\}.
\]
The union bound gives
\[
\Pr(\widehat\pi_{\rm H}\leq\pi)\geq1-\delta_\pi.
\]
Thus, the same conditional and unconditional coverage statements apply when Algorithm 1 uses $\widehat\pi_{\rm H}$. The Hoeffding bound is easier to compute but is generally smaller than the Clopper--Pearson bound in moderate samples. 

\newpage
\subsection{Full observed distribution Identification}

Propositions \ref{prop:ident} and \ref{prop:sharp} use only the reduced observables $(P,p_0,p_1)$. The full observed distribution further identifies the covariate distribution, covariate specific selection probabilities, and outcome distributions within the two selected treatment status. These additional objects reduce the identification region for the always-selected treated distribution.

We maintain Assumptions \ref{ass:random}--\ref{ass:positivity} and the standard Borel conditions on $\mathcal X$ and $\mathcal Y$. Let $F_X=\mathcal F(X)$ be the marginal covariate distribution. For every unit $i$, unit-level random assignment implied by Assumption \ref{ass:random} identifies
\[
s_d(x)
=\Pr\{S_i(d)=1\mid X_i=x\}
=\Pr(S_i=1\mid D_i=d,X_i=x),
\qquad d\in\{0,1\}.
\]
Because the potential outcome vectors have a common distribution, the left side is independent of $i$, and we suppress the unit index.

Assumption \ref{ass:monosel} implies $0\leq s_0(x)\leq s_1(x)\leq1$ for $F_X$-almost every $x$, and
\[
p_d=\int s_d(x)\,F_X(dx), \qquad d\in\{0,1\}.
\]
On $\{s_1(x)>0\}$, define the conditional share of always-selected units among treated-selected units with covariates $X=x$:
\[
\pi(x)=\frac{s_0(x)}{s_1(x)}
=\Pr\{\mathcal A\mid S(1)=1,X=x\},
\]
where $\pi(x)=0$ when $s_1(x)=0$.

The three principal strata are
\[
\mathcal A=\{S(0)=S(1)=1\},
\qquad
\mathcal M=\{S(0)=0,S(1)=1\},
\qquad
\mathcal N=\{S(0)=S(1)=0\},
\]
where $\mathcal A$ is the always-selected stratum, $\mathcal M$ is the marginal-in stratum, and $\mathcal N$ is the never-selected stratum. Selected controls belong to $\mathcal A$. Hence, their covariate distribution identifies the target covariate marginal:
\[
H_X
=\mathcal F(X\mid D=0,S=1)
=\mathcal F(X\mid\mathcal A),
\qquad
H_X(dx)=\frac{s_0(x)}{p_0}F_X(dx).
\]
The treated-selected covariate distribution is
\[
P_X
=\mathcal F(X\mid D=1,S=1),
\qquad
P_X(dx)=\frac{s_1(x)}{p_1}F_X(dx).
\]
Thus, $H_X\ll P_X$ and
\[
\frac{dH_X}{dP_X}(x)=\frac{\pi(x)}{\pi}
\quad P_X\text{-almost surely}.
\]

Because $\mathcal X$ and $\mathcal Y$ are standard Borel spaces, regular conditional distributions exist. Let
\[
P_x
=\mathcal F\{Y(1)\mid S(1)=1,X=x\}
=\mathcal F(Y\mid D=1,S=1,X=x)
\]
and
\[
G_x
=\mathcal F\{Y(0)\mid\mathcal A,X=x\}
=\mathcal F(Y\mid D=0,S=1,X=x).
\]
be the conditional treated-selected distribution of $Y(1)$ and the conditional distribution of $Y(0)$ among selected controls, respectively. Both distributions are identified from the full observed distribution. The unidentified conditional target distribution is
\[
Q_x=\mathcal F\{Y(1)\mid\mathcal A,X=x\}.
\]
For $H_X$-almost every $x$, $s_0(x)>0$ and therefore $\pi(x)>0$. Conditional on such an $x$, treated-selected units are a mixture of always-selected and marginal-in units. Hence,
\[
P_x=\pi(x)Q_x+\{1-\pi(x)\}R_x
\]
for some residual distribution 
\[
R_x = \mathcal F (Y(1) \mid \mathcal M, X=x),
\]
which implies the conditional likelihood ratio restriction
\[
Q_x\ll P_x,
\qquad
0\leq\frac{dQ_x}{dP_x}(y)\leq\frac{1}{\pi(x)}
\quad P_x\text{-almost surely}.
\]

For a candidate target distribution, write
\[
    Q(dx,dy) = H_X(dx) Q_x(dy),
\]
and define the full observed distribution identification class
\[
\mathcal Q_{\rm full}
=\Bigl\{H_X(dx)Q_x(dy):\;Q_x\ll P_x, 
0\leq\frac{dQ_x}{dP_x}(y)\leq\frac{1}{\pi(x)} \quad P_x\text{-a.s. for }H_X\text{-a.e. }x\Bigr\}.
\]
Proposition \ref{prop:full-equivalence} shows that the $\mathcal Q_{\rm full}$ can be written as the Lee ambiguity set based on the reduced observables intersected with the identified covariate marginal restriction.

\begin{prop}[Equivalent joint and conditional representations]
\label{prop:full-equivalence}
Let $P(dx,dy)=P_X(dx)P_x(dy)$. Then
\[
\mathcal Q_{\rm full}
=\{Q\in\mathcal Q(P,\pi):Q_X=H_X\}.
\]
If $Q(dx,dy)=H_X(dx)Q_x(dy)\in\mathcal Q_{\rm full}$, choose a jointly measurable version $q_x=dQ_x/dP_x$ on the set $\{dH_X/dP_X>0\}$ and define the product below as zero on the set $\{dH_X/dP_X=0\}$. Then a version of the joint density is
\[
\frac{dQ}{dP}(x,y)=\frac{\pi(x)}\pi q_x(y)
\quad P\text{-almost surely}.
\]
Conversely, if $Q\in\mathcal Q(P,\pi)$ and $Q_X=H_X$, then for $H_X$-almost every $x$,
\[
\frac{dQ_x}{dP_x}(y)
=\frac{\pi}{\pi(x)}\frac{dQ}{dP}(x,y)
\quad P_x\text{-almost surely}.
\]
\end{prop}

Therefore, the full observed distribution separates an identified covariate shift from $P_X$ to $H_X$ and a partially identified conditional outcome shift from $P_x$ to $Q_x$. In consequence, $\mathcal Q_{\rm full}$ is no larger than the reduced-information class $\mathcal Q(P,\pi)$ and can support sharper prediction rules.

Propositions \ref{prop:full-likelihood} and \ref{prop:full-sharpness} are the full observed distribution analogue of Propositions \ref{prop:ident} and \ref{prop:sharp}. Together with Proposition \ref{prop:full-likelihood}, the result in Proposition  \ref{prop:full-sharpness} implies that $\mathcal Q_{\rm full}$ is the sharp identification region given the full observed distribution.

\begin{prop}[Full observed distribution likelihood ratio domination]
\label{prop:full-likelihood}
Suppose Assumptions \ref{ass:random}--\ref{ass:positivity} hold. For any data-generating process matching the full observed distribution, let
\[
Q_0=\mathcal F\{X,Y(1)\mid\mathcal A\}.
\]
Then $Q_0\in\mathcal Q_{\rm full}$. Equivalently, if
\[
Q_0(dx,dy)=H_X(dx)Q_{0,x}(dy),
\]
then, for $H_X$-almost every $x$,
\[
Q_{0,x}\ll P_x,
\qquad
0\leq\frac{dQ_{0,x}}{dP_x}(y)\leq\frac{1}{\pi(x)}
\quad P_x\text{-almost surely}.
\]
If $\pi(x)<1$, there exists a probability distribution $R_{0,x}$ such that
\[
P_x=\pi(x)Q_{0,x}+\{1-\pi(x)\}R_{0,x}.
\]
If $\pi(x)=1$, then $Q_{0,x}=P_x$.
\end{prop}

\begin{prop}[Sharpness of $\mathcal Q_{\rm full}$]
\label{prop:full-sharpness}
For every
\[
\widetilde Q(dx,dy)=H_X(dx)\widetilde Q_x(dy)
\in\mathcal Q_{\rm full},
\]
there exists a data-generating process for $(Y(1),Y(0),S(1),S(0),X)$ that satisfies monotone selection, matches
\[
F_X,\qquad s_0(x),\qquad s_1(x),\qquad P_x,\qquad G_x,
\]
and satisfies
\[
\mathcal F\{X,Y(1)\mid\mathcal A\}=\widetilde Q.
\]
Any prescribed joint assignment-vector distribution with positive unit-specific margins may be appended independently of independent copies of the latent construction. The resulting observed process satisfies Assumption \ref{ass:random} and matches that assignment distribution.
\end{prop}

\newpage
\subsection{Proofs}

\begin{proof}[\textbf{Proof of Proposition \ref{prop:ident}}]

The measure theoretic arguments use the Radon--Nikodym theorem and disintegration results developed in \citet{kallenberg1997foundations}. Let
\[
Z=(X,Y(1)),\qquad
P=\mathcal F\{Z\mid S(1)=1\},\qquad
Q_0=\mathcal F\{Z\mid S(0)=1\}.
\]
Also, let $p_d=\Pr\{S(d)=1\}$ for $d\in\{0,1\}$ and $\pi=p_0/p_1$. Assumption \ref{ass:positivity} gives $p_0>0$, and Assumption \ref{ass:monosel} gives $0<p_0\leq p_1$. Hence, $\pi\in(0,1]$.

For every $B\in\mathcal B$,
\[
\begin{aligned}
\pi Q_0(B)
&=\frac{p_0}{p_1}\Pr(Z\in B\mid S(0)=1)\\
&=\frac{1}{p_1}\Pr\{Z\in B,S(0)=1\}\\
&\leq\frac{1}{p_1}\Pr\{Z\in B,S(1)=1\}\\
&=\Pr(Z\in B\mid S(1)=1)\\
&=P(B).
\end{aligned}
\]
The inequality follows because monotone selection implies $\{S(0)=1\}\subseteq\{S(1)=1\}$ almost surely. Thus,
\begin{equation}
\label{prop1:eqn1}
\pi Q_0\leq P
\end{equation}
as finite measures.

To establish absolute continuity, take any $B\in\mathcal B$ with $P(B)=0$.  \eqref{prop1:eqn1} gives $0\leq\pi Q_0(B)\leq P(B)=0$. Since $\pi>0$, $Q_0(B)=0$, and therefore $Q_0\ll P$. Let
\[
w=\frac{dQ_0}{dP}.
\]
Because $Q_0$ is a probability measure, $w\geq0$ $P$-almost surely and $\int_{\mathcal Z}w\,dP=1$.

Define the finite measure $\mu=P-\pi Q_0$. \eqref{prop1:eqn1} shows that $\mu$ is nonnegative, and for every $B \in \mathcal B$,
\[
\mu(B)
=P(B)-\pi Q_0(B)
=\int_B\{1-\pi w(z)\}\,P(dz).
\]
The Radon--Nikodym derivative of a nonnegative measure is nonnegative. Hence,
\[
1-\pi w(z)\geq0
\quad P\text{-almost surely},
\]
which implies
\[
0\leq\frac{dQ_0}{dP}(z)=w(z)\leq\frac{1}{\pi}
\quad P\text{-almost surely}.
\]
Thus, $Q_0\in\mathcal Q(P,\pi)$.

It remains to establish the mixture representation. If $\pi=1$, then $0\leq w\leq1$ and $\int w\,dP=1$. Therefore, $\int(1-w)\,dP=0$, and the nonnegativity of $1-w$ implies $w=1$ $P$-almost surely. Hence, $Q_0=P$.

Suppose instead that $\pi<1$. The measure $\mu=P-\pi Q_0$ has total mass $1-\pi$. Define
\[
R_{\rm res}(B)
=\frac{\mu(B)}{1-\pi}
=\frac{P(B)-\pi Q_0(B)}{1-\pi},
\qquad B\in\mathcal B.
\]
Then $R_{\rm res}$ is a probability measure and
\[
P=\pi Q_0+(1-\pi)R_{\rm res}.
\]
\end{proof}

\begin{proof}[\textbf{Proof of Proposition \ref{prop:sharp}}]

Fix $P\in\mathcal P(\mathcal Z,\mathcal B)$, $0<p_0\leq p_1\leq1$, and $\pi=p_0/p_1$. Let $\widetilde Q\in\mathcal Q(P,\pi)$.

\paragraph{Case 1. $\pi<1$.} 

Let $w=d\widetilde Q/dP$. For $B\in\mathcal B$, define
\[
\mu(B)
=P(B)-\pi\widetilde Q(B)
=\int_B\{1-\pi w(z)\}\,P(dz).
\]
Since $0\leq w\leq1/\pi$ $P$-almost surely, $\mu$ is nonnegative. Its total mass is $1-\pi$, so
\[
R(B)=\frac{P(B)-\pi\widetilde Q(B)}{1-\pi},
\qquad B\in\mathcal B,
\]
defines a probability measure. By construction,
\begin{equation}
\label{prop2:eqn1}
P=\pi\widetilde Q+(1-\pi)R.
\end{equation}

We now construct the latent population. Choose any $z_N\in\mathcal Z$. Define a joint distribution for the principal-stratum variable $T\in\{\mathcal A,\mathcal M,\mathcal N\}$ and $Z=(X,Y(1))$ by
\[
\begin{aligned}
\Pr(T=\mathcal A,Z\in B)&=p_0\widetilde Q(B),\\
\Pr(T=\mathcal M,Z\in B)&=(p_1-p_0)R(B),\\
\Pr(T=\mathcal N,Z\in B)&=(1-p_1)\delta_{z_N}(B),
\end{aligned}
\qquad B\in\mathcal B,
\]
where $\delta_{z_N}$ is the point mass at $z_N$. In particular, 
\[
\Pr(T = \mathcal A)= p_0, \qquad \Pr(T = \mathcal M)= p_1 - p_0, \qquad \Pr(T = \mathcal N)= 1 - p_1.
\]
The three masses are nonnegative and sum to one. Define the potential selection indicators
\[
(S(0),S(1))=
\begin{cases}
(1,1), & T=\mathcal A,\\
(0,1), & T=\mathcal M,\\
(0,0), & T=\mathcal N.
\end{cases}
\]
Then $S(1)\geq S(0)$ almost surely and
\[
\Pr\{S(0)=1\}=p_0,
\qquad
\Pr\{S(1)=1\}=p_1.
\]
Hence, Assumptions \ref{ass:monosel} and \ref{ass:positivity} hold. 

Choose any fixed $y_0\in\mathcal Y$ and set $Y(0)=y_0$. We first verify the two latent distributions. For every $B\in\mathcal B$,
\[
\begin{aligned}
\Pr(Z\in B\mid S(1)=1)
&=\frac{p_0\widetilde Q(B)+(p_1-p_0)R(B)}{p_1}\\
&=\pi\widetilde Q(B)+(1-\pi)R(B)\\
&=P(B),
\end{aligned}
\]
where the last equality follows from \eqref{prop2:eqn1}. Hence,
\[
\mathcal F\{X,Y(1)\mid S(1)=1\}=P.
\]
Moreover, $S(0)=1$ if and only if $T=\mathcal A$, so
\[
\mathcal F\{X,Y(1)\mid S(0)=1\}
=\mathcal F(Z\mid T=\mathcal A)
=\widetilde Q.
\]
Let $D$ be any binary variable that is independent of the constructed potential outcome vector and satisfies $0<\Pr(D=1)<1$. Define
\[
S=DS(1)+(1-D)S(0),
\qquad
Y=DY(1)+(1-D)Y(0).
\]
Unit-level random assignment and consistency then give
\[
\mathcal F(X,Y\mid D=1,S=1)
=\mathcal F\{X,Y(1)\mid S(1)=1\}
=P.
\]
For a sample, independent copies of the latent construction can be combined with any prescribed distribution for $\mathbf D_n$ having positive unit-specific margins by appending that distribution independently. Hence, Assumption \ref{ass:random} is satisfied. This proves attainability when $\pi<1$.

\paragraph{Case 2. $\pi=1$.}
In this case, $p_0=p_1$. Since $\widetilde Q\in\mathcal Q(P,1)$, its density $w=d\widetilde Q/dP$ satisfies $0\leq w\leq1$ and $\int w\,dP=1$. Thus, $w=1$ $P$-almost surely, and
\begin{equation}
\label{prop2:eqn2}
\widetilde Q=P.
\end{equation}

Choose any $z_N\in\mathcal Z$ and define
\[
\Pr(T=\mathcal A,Z\in B)=p_1P(B),
\qquad
\Pr(T=\mathcal N,Z\in B)=(1-p_1)\delta_{z_N}(B).
\]
Set $(S(0),S(1))=(1,1)$ when $T=\mathcal A$ and $(S(0),S(1))=(0,0)$ when $T=\mathcal N$, and choose a fixed $Y(0)$. Then
\[
\mathcal F(Z\mid S(1)=1)
=\mathcal F(Z\mid S(0)=1)
=P=\widetilde Q.
\]
Appending any binary assignment variable independent of the latent construction, with assignment probability in $(0,1)$, and applying consistency gives
\[
\mathcal F(X,Y\mid D=1,S=1)=P.
\]
As in the first case, any prescribed distribution for $\mathbf D_n$ with positive unit-specific margins may be appended independently in a sample. Therefore, the construction matches the required observables and target distribution.

Proposition \ref{prop:ident} gives necessity. The two constructions give sufficiency because every $\widetilde Q\in\mathcal Q(P,\pi)$ is attainable while $(P,p_0,p_1)$ remains fixed. Hence, $\mathcal Q(P,\pi)$ is the sharp identification region based on the reduced observables.
\end{proof}

\begin{proof}[\textbf{Proof of Lemma \ref{lem:calibration-design}}]
Let $P_W = \mathcal F(W_1)$ denote the distribution of potential outcome vector. $\mathbf D_n\indep\mathbf W_n$ by Assumption \ref{ass:random} and $\mathcal F(\mathbf W_n)=P_W^{\otimes n}$ because $W_1, \ldots, W_n$ are independent and identically distributed. Hence,
\[
\mathcal F(\mathbf W_n\mid\mathbf D_n)=P_W^{\otimes n}
\qquad\text{almost surely}.
\]
For $d,s\in\{0,1\}$, define the conditional kernel
\[
K_{d,s}
=\mathcal F\big(W_i\mid S_i(d)=s\big),
\]
using an arbitrary version when $\Pr\{S_i(d)=s\}=0$. Since
\[
S_i=S_i(D_i)
=D_iS_i(1)+(1-D_i)S_i(0)
\]
is a coordinatewise function of $(D_i,W_i)$, conditioning the preceding product distribution on $\mathbf S_n = (S_1, \ldots, S_n)$ gives
\begin{equation}
\label{eq:calibration-product-W}
\mathcal F(\mathbf W_n\mid\mathbf D_n,\mathbf S_n)
=
\bigotimes_{i=1}^n K_{D_i,S_i}
\qquad\text{almost surely}.
\end{equation}
Since $R_{\rm sp}\indep(\mathbf D_n,\mathbf W_n)$, (\ref{eq:calibration-product-W}) remains valid after conditioning on the sigma-field $\mathcal G_n=\sigma(\mathbf D_n,\mathbf S_n,R_{\rm sp})$.

Fix $r\in\{0,\ldots,n\}$ with $\Pr(m=r)>0$. On $\{m=r\}$, the indices $i_1<\cdots<i_r$ are $\mathcal G_n$-measurable. For every $j\leq r$, $i_j\in I_{\rm cal}\subseteq I_1$, so $D_{i_j}=S_{i_j}=1$. By consistency and the common distribution of the $W_i$'s,
\[
\mathcal F(Z_j\mid\mathcal G_n)
=
\mathcal F\big((X_{i_j},Y_{i_j}(1))\mid S_{i_j}(1)=1\big)
=
\mathcal F\big((X,Y(1))\mid S(1)=1\big)
=P.
\]
The product kernel in \eqref{eq:calibration-product-W} also yields conditional independence across the coordinates indexed by $I_{\rm cal}$. Hence
\begin{equation}
\label{eq:calibration-product-Z}
\mathcal F\big((Z_1,\ldots,Z_r)\mid\mathcal G_n\big)
=P^{\otimes r}
\qquad\text{almost surely on }\{m=r\}.
\end{equation}

Because $I_{\rm train}\cap I_{\rm cal}=\varnothing$, the same product kernel implies
\[
\sigma(Z_1,\ldots,Z_r)
\indep
\mathcal T_n
\mid\mathcal G_n
\qquad\text{on }\{m=r\},
\]
where $\mathcal T_n = \sigma((X_i,Y_i)_{i \in I_{\rm train}})$. The restriction
$R_{\rm fit}\indep(\mathbf D_n,\mathbf W_n,R_{\rm sp})$
then gives
\[
\sigma(Z_1,\ldots,Z_r)
\indep
\mathcal T_n\vee\sigma(R_{\rm fit})
\mid\mathcal G_n
\qquad\text{on }\{m=r\}.
\]
Since $\mathcal H=\mathcal G_n\vee\mathcal T_n\vee\sigma(R_{\rm fit})$,  (\ref{eq:calibration-product-Z}) implies
\[
\mathcal F\big((Z_1,\ldots,Z_r)\mid\mathcal H\big)
=P^{\otimes r}
\qquad\text{almost surely on }\{m=r\}.
\]
\end{proof}

\begin{proof}[\textbf{Proof of Theorem \ref{thm1:coverage}}]

Fix an arbitrary $Q\in\mathcal Q(P,\pi)$. By Lemma \ref{lem:calibration-design} and the target-sampling condition in the theorem, the regular conditional distribution under $\mathbb P$, given $\mathcal H$ and each event $\{m=r\}$, is $P^{\otimes r}\otimes Q$. The score function $V$, the calibration size $m$, the share used for calibration $\widehat\pi$, and the index $k$ are fixed conditional on $\mathcal H$. Thus the calibration observations are independent draws from $P$, while the test observation $Z_{m+1}=(X_{m+1},Y_{m+1})$ is independent of the calibration sample and has distribution $Q$. We suppress the event $\{m=r\}$ in the notation below.

Let $V_{m+1}=V(Z_{m+1})$. By definition of the prediction set,
\[
\{Y_{m+1}\notin\widehat C_1(X_{m+1})\}
=
\{V_{m+1}>\widehat t\}.
\]
For $t\in\overline{\mathbb R}$, define
\[
A_t=\{z\in\mathcal Z:V(z)>t\},
\qquad A_{+\infty}=\varnothing.
\]
Let $\mathcal C=\sigma(Z_1,\ldots,Z_m)$. The threshold $\widehat t=V_{(k)}$ is $\mathcal H\vee\mathcal C$ measurable, and the test observation is conditionally independent of $\mathcal C$ given $\mathcal H$. Therefore,
\[
\begin{aligned}
\mathbb P(V_{m+1}>\widehat t\mid\mathcal H)
&=\mathbb E_{\mathbb P}
\left[
\mathbb P(V_{m+1}>\widehat t\mid\mathcal H,\mathcal C)
\mid\mathcal H
\right]\\
&=\mathbb E_{\mathbb P}\{Q(A_{\widehat t})\mid\mathcal H\}.
\end{aligned}
\]

The likelihood ratio restriction gives, for every $B\in\mathcal B$,
\[
Q(B)\leq\frac{1}{\pi}P(B).
\]
Applying this inequality to $A_{\widehat t}$ and taking conditional expectations yields
\begin{equation}
\label{thm1:EPAt}
\mathbb P(V_{m+1}>\widehat t\mid\mathcal H)
\leq
\frac{1}{\pi}
\mathbb E_{\mathbb P}\{P(A_{\widehat t})\mid\mathcal H\}.
\end{equation}

Introduce an auxiliary observation $Z_{m+1}^{\circ}\sim P$ that is conditionally independent of $Z_1,\ldots,Z_m$ given $\mathcal H$, and define $V_{m+1}^{\circ}=V(Z_{m+1}^{\circ})$. Conditional on $\mathcal H$ and the calibration sample, the threshold is fixed. Hence,
\begin{equation}
\label{thm1:Vm+1circ}
\mathbb E_{\mathbb P}\{P(A_{\widehat t})\mid\mathcal H\}
=
\mathbb P(V_{m+1}^{\circ}>\widehat t\mid\mathcal H)
=
\mathbb P(V_{m+1}^{\circ}>V_{(k)}\mid\mathcal H).
\end{equation}
Combining \eqref{thm1:EPAt} and \eqref{thm1:Vm+1circ} gives
\begin{equation}
\label{thm1:combining1}
\mathbb P(V_{m+1}>\widehat t\mid\mathcal H)
\leq
\frac{1}{\pi}
\mathbb P(V_{m+1}^{\circ}>V_{(k)}\mid\mathcal H).
\end{equation}

We next handle ties. The rank step
is the standard split-conformal argument
(see, e.g., \citet{vovk2005algorithmic,lei2018distribution}). Let $U_1,\ldots,U_m,U_{m+1}$ be independent $\operatorname{Unif}(0,1)$ variables that are independent of all other variables. Define
\[
W_i=(V_i,U_i),\quad i=1,\ldots,m,
\qquad
W_{m+1}=(V_{m+1}^{\circ},U_{m+1}),
\]
and order these pairs using lexicographic order:
\[
(v,u) <_{\text{lex}}(v',u') \iff [v < v'] \text{ or } [v = v' \text{ and } u < u'].
\]
Conditional on $\mathcal H$, the $m+1$ augmented scores are exchangeable and almost surely distinct. Therefore, their lexicographic rank
\[
\mathcal R
=1+\sum_{i=1}^m\mathbf 1\{W_i<_{\rm lex}W_{m+1}\}
\]
is uniform:
\begin{equation}
\label{thm1:uniform_rank}
\mathbb P(\mathcal R=r\mid\mathcal H)=\frac{1}{m+1},
\qquad r=1,\ldots,m+1.
\end{equation}

If $V_{m+1}^{\circ}>V_{(k)}$, then at least $k$ calibration scores satisfy $V_i\leq V_{(k)}<V_{m+1}^{\circ}$. Each corresponding augmented score is lexicographically smaller than $W_{m+1}$. Thus,
\begin{equation}
\label{thm1:subset}
\{V_{m+1}^{\circ}>V_{(k)}\}
\subseteq
\{\mathcal R\geq k+1\}.
\end{equation}
The inclusion also holds when $k=m+1$, since the event on the left is then empty. Equations \eqref{thm1:uniform_rank} and \eqref{thm1:subset} imply
\begin{equation}
\label{thm1:combining2}
\mathbb P(V_{m+1}^{\circ}>V_{(k)}\mid\mathcal H)
\leq
\frac{m+1-k}{m+1}.
\end{equation}
Combining \eqref{thm1:combining1} and \eqref{thm1:combining2} gives
\[
\mathbb P\{Y_{m+1}\notin\widehat C_1(X_{m+1})\mid\mathcal H\}
\leq
\frac{m+1-k}{(m+1)\pi}.
\]
Therefore,
\begin{equation}
\label{thm1:combining3}
\mathbb P\{Y_{m+1}\in\widehat C_1(X_{m+1})\mid\mathcal H\}
\geq
1-\frac{m+1-k}{(m+1)\pi}
\geq
1-\alpha-\widehat\Delta.
\end{equation}

It remains to bound $\widehat\Delta$. Since
\[
k=\left\lceil(m+1)(1-\alpha\widehat\pi)\right\rceil
\geq
(m+1)(1-\alpha\widehat\pi),
\]
we have
\[
\frac{m+1-k}{(m+1)\pi}
\leq
\alpha\frac{\widehat\pi}{\pi}.
\]
Consequently,
\[
\widehat\Delta
\leq
\left[\alpha\frac{\widehat\pi}{\pi}-\alpha\right]_+
=
\alpha\left(\frac{\widehat\pi}{\pi}-1\right)_+
=
\frac{\alpha}{\pi}(\widehat\pi-\pi)_+.
\]
In particular, $\widehat\Delta=0$ whenever $\widehat\pi\leq\pi$.
\end{proof}

\begin{proof}[\textbf{Proof of Corollary
\ref{coro:plugin-asymptotic-coverage}}]
For each $n$ and every $Q\in\mathcal Q(P,\pi)$, Theorem
\ref{thm1:coverage} gives
\[
\mathbb P_{n,Q}\!
\left\{
Y_{n,m_n+1}\in
\widehat C_{1,n}(X_{n,m_n+1})
\mid\mathcal H_n
\right\}
\geq 1-\alpha-\widehat\Delta_n,
\]
where
\[
0\leq\widehat\Delta_n
\leq
\frac{\alpha}{\pi}(\widehat\pi_n-\pi)_+.
\]
Since $\widehat\pi_n \overset{p}{\longrightarrow} \pi$,
\[
    \frac{\alpha}{\pi}(\widehat \pi_n - \pi)_+ \overset{p}{\longrightarrow} 0.
\]
This proves the conditional shortfall statement. Since $0\leq\widehat\pi_n\leq1$, the error satisfies
$0\leq\widehat\Delta_n\leq\alpha/\pi$. For every $\epsilon>0$,
\[
E(\widehat\Delta_n)
\leq\epsilon+\frac{\alpha}{\pi}
\Pr(\widehat\Delta_n>\epsilon).
\]
Hence, convergence in probability implies
$E(\widehat\Delta_n)\to0$. This expectation is common to all target distributions under the experiment specified above.  Taking expectations in the finite-sample inequality gives, uniformly over $Q\in\mathcal Q(P,\pi)$,
\[
\mathbb P_{n,Q}\!
\left\{
Y_{n,m_n+1}\in
\widehat C_{1,n}(X_{n,m_n+1})
\right\}
\geq 1-\alpha-E(\widehat\Delta_n).
\]
Taking the infimum over $Q$ and then the lower limit proves the result.
\end{proof}

\begin{proof}[\textbf{Proof of Corollary \ref{coro:iteband}}]
For any set $C$ and scalars $a$ and $b$, $a-b\in C-b$ if and only if $a\in C$. Applying this identity pathwise gives
\[
\begin{aligned}
\left\{\tau_{m+1}\in
\widehat C_\tau\bigl(X_{m+1},Y_{m+1}(0)\bigr)\right\}
&=
\left\{Y_{m+1}(1)-Y_{m+1}(0)
\in
\widehat C_1(X_{m+1})-Y_{m+1}(0)\right\}\\
&=
\left\{Y_{m+1}(1)\in\widehat C_1(X_{m+1})\right\}.
\end{aligned}
\]
This equality does not require independence between $Y(0)$ and $Y(1)$.

Under $\mathbb P_\Gamma$, the calibration observations have conditional distribution $P^{\otimes m}$, and the independent test unit has joint potential-outcome distribution $\Gamma$. Its $(X,Y(1))$ marginal is $Q_\Gamma\in\mathcal Q(P,\pi)$. Theorem \ref{thm1:coverage} therefore gives
\[
\mathbb P_\Gamma
\left\{
\tau_{m+1}\in
\widehat C_\tau\bigl(X_{m+1},Y_{m+1}(0)\bigr)
\mid\mathcal H
\right\}
\geq
1-\alpha-\widehat\Delta.
\]

It remains to connect $\Gamma$ to selected controls. Let $W=(X,Y(0),Y(1))$. For a control unit, $S=S(0)$. Monotone selection implies $\{S(0)=1\}=\mathcal A$ almost surely, and unit-level random assignment implied by Assumption \ref{ass:random} gives $D\indep(W,S(0),S(1))$. Hence,

\[
\begin{aligned}
\mathcal F(W\mid D=0,S=1)
&=\mathcal F(W\mid S(0)=1)\\
&=\mathcal F(W\mid\mathcal A)\\
&=\Gamma.
\end{aligned}
\]
Thus, an independent selected-control draw has distribution $\Gamma$, and its untreated outcome $Y(0)$ is observed. The shifted set consequently has the stated marginal ITE coverage.
\end{proof}

\begin{proof}[\textbf{Proof of Proposition \ref{prop:sharpscore}}]
Fix $t\in\mathbb R$ and define
\[
A_t=\{z\in\mathcal Z:V(z)>t\},
\qquad
a=P(A_t).
\]
Because $V$ is measurable, $A_t\in\mathcal B$. For any $Q\in\mathcal Q(P,\pi)$, let $w=dQ/dP$. Then
\[
Q(A_t)
=\int_{A_t}w\,dP
\leq
\frac{1}{\pi}P(A_t)
=\frac{a}{\pi}.
\]
Since $Q(A_t)\leq1$,
\begin{equation}
\label{prop3:eqn1}
\sup_{Q\in\mathcal Q(P,\pi)}Q(A_t)
\leq
\min\left\{1,\frac{a}{\pi}\right\}.
\end{equation}
We show that this upper bound is attainable.

\paragraph{Case 1. $a\leq\pi$.}
Suppose first that $a<1$. Define
\[
c=\frac{1-a/\pi}{1-a}
\]
and
\[
w_t(z)
=\frac{1}{\pi}\mathbbm{1}_{A_t}(z)
+c\mathbbm{1}_{A_t^c}(z).
\]
The inequality $a\leq\pi$ gives $c\geq0$, and
\[
\frac{1}{\pi}-c
=\frac{1-\pi}{\pi(1-a)}
\geq0.
\]
Hence, $0\leq w_t\leq1/\pi$ $P$-almost surely. Moreover,
\[
\int_{\mathcal Z}w_t\,dP
=\frac{a}{\pi}+c(1-a)
=1.
\]
Define $Q_t$ by $Q_t(B)=\int_Bw_t\,dP$. Then $Q_t\in\mathcal Q(P,\pi)$ and
\[
Q_t(A_t)=\frac{a}{\pi}
=\min\left\{1,\frac{a}{\pi}\right\}.
\]
If $a=1$ and $a\leq\pi\leq1$, then $\pi=1$. Taking $Q_t=P$ gives the same equality.

\paragraph{Case 2. $a>\pi$.}
In this case, $a>0$. Define
\[
Q_t(B)
=P(B\mid A_t)
=\frac{P(B\cap A_t)}{a},
\qquad B\in\mathcal B.
\]
Its Radon--Nikodym derivative is $dQ_t/dP=\mathbbm{1}_{A_t}/a$. Since $a>\pi$,
\[
0\leq\frac{dQ_t}{dP}\leq\frac{1}{\pi}
\quad P\text{-almost surely}.
\]
Thus, $Q_t\in\mathcal Q(P,\pi)$, and $Q_t(A_t)=1$. Because $a/\pi>1$,
\[
Q_t(A_t)
=1
=\min\left\{1,\frac{a}{\pi}\right\}.
\]

Both cases attain the upper bound in \eqref{prop3:eqn1}. Therefore,
\[
\sup_{Q\in\mathcal Q(P,\pi)}Q(V>t)
=
\min\left\{1,\frac{P(V>t)}{\pi}\right\}.
\]
\end{proof}

\begin{proof}[\textbf{Proof of Proposition \ref{prop:sharpthreshold}}]
Fix $t\in\mathbb R$. Proposition \ref{prop:sharpscore} implies
\begin{equation}
\label{prop4:eqn1}
\inf_{Q\in\mathcal Q(P,\pi)}Q(V\leq t)
=
\max\left\{0,1-\frac{P(V>t)}{\pi}\right\}.
\end{equation}
The attaining distribution in Proposition \ref{prop:sharpscore} also attains the infimum in \eqref{prop4:eqn1}.

Because $1-\alpha>0$, equation \eqref{prop4:eqn1} gives
\[
\inf_{Q\in\mathcal Q(P,\pi)}Q(V\leq t)\geq1-\alpha
\quad\Longleftrightarrow\quad
P(V>t)\leq\alpha\pi.
\]
Define
\[
\mathcal F_\alpha
=\{t\in\mathbb R:P(V>t)\leq\alpha\pi\}.
\]
Since $V$ is real valued and $0<\alpha\pi<1$, $\mathcal F_\alpha$ is nonempty and bounded below. Let $t^*=\inf\mathcal F_\alpha$.

Write $g(t)=P(V>t)$. The function $g$ is nonincreasing and right continuous. Since $\mathcal F_\alpha$ is an upper set, $t^*+1/n\in\mathcal F_\alpha$ for every $n\geq1$. Therefore,
\[
g(t^*)
=\lim_{n\to\infty}g(t^*+1/n)
\leq\alpha\pi,
\]
so $t^*\in\mathcal F_\alpha$. For every $t<t^*$, by contrast, $P(V>t)>\alpha\pi$.

Let $F(t)=P(V\leq t)$. Since $P(V>t)=1-F(t)$,
\[
\begin{aligned}
t^*
&=\inf\{t:P(V>t)\leq\alpha\pi\}\\
&=\inf\{t:F(t)\geq1-\alpha\pi\}\\
&=F^{-1}(1-\alpha\pi).
\end{aligned}
\]
Thus, $t^*$ is the lower $(1-\alpha\pi)$ quantile, including when $F$ has atoms.

Finally, take any $t<t^*$. Then $P(V>t)/\pi>\alpha$, and Proposition \ref{prop:sharpscore} provides $Q_t\in\mathcal Q(P,\pi)$ such that
\[
Q_t(V>t)
=\min\left\{1,\frac{P(V>t)}{\pi}\right\}
>\alpha.
\]
Equivalently, $Q_t(V\leq t)<1-\alpha$. Therefore, no threshold below $t^*$ can guarantee coverage uniformly over $\mathcal Q(P,\pi)$. This proves the minimax claim.
\end{proof}

\begin{proof}[\textbf{Proof of Proposition
\ref{prop:threshold-consistency}}]
The quantile argument follows the inversion logic discussed in \citet{van2000asymptotic}. Conditional on $\mathcal H_n$, we use \citet{hoeffding1963probability} at two
$\mathcal H_n$-measurable evaluation points. 

Because $\widehat\pi_n \overset{p}{\longrightarrow} \pi>0$ and $m_n \overset{p}{\longrightarrow} \infty$,
\[
\Pr\!\left\{
\alpha\widehat\pi_n\geq\frac{1}{m_n+1}
\right\}
\longrightarrow1.
\]
On this event,
$(m_n+1)(1-\alpha\widehat\pi_n)\leq m_n$, so $k_n\leq m_n$.
Hence $\Pr(k_n=m_n+1) \longrightarrow 0$.  On $\{k_n\leq m_n\}$, define $u_n=k_n/m_n$.  The definition of the ceiling function gives
\[
(m_n+1)(1-\alpha\widehat\pi_n)
\leq k_n
<
(m_n+1)(1-\alpha\widehat\pi_n)+1.
\]
Consequently,
\[
\left|u_n-(1-\alpha\widehat\pi_n)\right|
\leq\frac{2}{m_n}.
\]
It follows that
\[
|u_n-q|
\leq
\alpha|\widehat\pi_n-\pi|+\frac{2}{m_n}
\overset{p}{\longrightarrow} 0.
\]

For $m_n\geq1$, define the empirical CDF
\[
\widehat F_n(t)
=
\frac{1}{m_n}
\sum_{i=1}^{m_n}\mathbbm 1\{V_{n,i}\leq t\}.
\]
The lower quantile $t_n^*$ is $\mathcal H_n$-measurable.  Consequently, for
any fixed $\varepsilon>0$, conditional Hoeffding bounds can be applied at the
random points $t_n^*\pm\varepsilon$.

For part (i), fix $\epsilon > 0$ and  write $\eta=\eta_\varepsilon$.  Conditional on $\mathcal H_n$,
Hoeffding's inequality gives
\[
\Pr\!\left(
\left|
\widehat F_n(t_n^*\pm\varepsilon)
-F_n(t_n^*\pm\varepsilon)
\right|>\frac{\eta}{3}
\,\middle|\,
\mathcal H_n
\right)
\leq
2\exp\!\left(-\frac{2m_n\eta^2}{9}\right).
\]
Since $m_n \overset{p}{\longrightarrow} \infty$, the expectation of the right-hand side converges to
zero.  Thus both empirical deviations are at most $\eta/3$ with probability
tending to one.

Now consider the event on which
\[
F_n(t_n^*-\varepsilon)\leq q-\eta,
\qquad
F_n(t_n^*+\varepsilon)\geq q+\eta,
\]
\[
|u_n-q|<\eta/3,
\qquad k_n\leq m_n,
\]
and both empirical deviations are at most
$\eta/3$.  If $\widehat t_n\leq t_n^*-\varepsilon$, then at least $k_n$
calibration scores are no larger than $t_n^*-\varepsilon$. Thus,
\[
\widehat F_n(t_n^*-\varepsilon)\geq u_n.
\]
On the same event, however,
\[
\widehat F_n(t_n^*-\varepsilon)
\leq q-\frac{2\eta}{3}
<q-\frac{\eta}{3}
<u_n,
\]
which is a contradiction.  Similarly, if
$\widehat t_n>t_n^*+\varepsilon$, then fewer than $k_n$ scores are no larger
than $t_n^*+\varepsilon$. Thus,
\[
\widehat F_n(t_n^*+\varepsilon)
\leq\frac{k_n-1}{m_n}<u_n.
\]
But the same event also gives
\[
\widehat F_n(t_n^*+\varepsilon)
\geq q+\frac{2\eta}{3}
>q+\frac{\eta}{3}
>u_n,
\]
which is a contradiction.  Therefore,
\[
\Pr\{|\widehat t_n-t_n^*|>\varepsilon\}\longrightarrow0.
\]

For part (ii), fix $\varepsilon>0$ and define
\[
\eta
=
\frac{1}{3}
\min\left\{
q-F(t^*-\varepsilon),
F(t^*+\varepsilon)-q
\right\}>0.
\]
Uniform convergence of $F_n$ to $F$ implies that, with probability tending to
one,
\[
F_n(t^*-\varepsilon)\leq q-2\eta,
\qquad
F_n(t^*+\varepsilon)\geq q+2\eta.
\]
By the definition of the lower quantile, on
this event,
\[
t^*-\varepsilon<t_n^*\leq t^*+\varepsilon.
\]
Therefore, $t_n^* \overset{p}{\longrightarrow} t^*$.  Conditional Hoeffding bounds at the fixed points
$t^*\pm\varepsilon$ imply that both empirical deviations are at most
$\eta/2$ with probability tending to one.  Intersecting this event with
$|u_n-q|<\eta$ gives
\[
\widehat F_n(t^*-\varepsilon)
\leq q-\frac{3\eta}{2}
<q-\eta
<u_n
\]
and
\[
\widehat F_n(t^*+\varepsilon)
\geq q+\frac{3\eta}{2}
>q+\eta
>u_n
\]
with probability tending to one.
The same order statistic arguments used in part (i) then imply
\[
t^*-\varepsilon<\widehat t_n\leq t^*+\varepsilon.
\]
Therefore, $\widehat t_n \overset{p}{\longrightarrow} t^*$.
\end{proof}

\begin{proof}[\textbf{Proof of Proposition
\ref{prop:plugin-pi-consistency}}]

Condition \ref{cond:selection-counts} gives
\[
M_{d,n}\mid(D_1,\ldots,D_n)
\sim
\operatorname{Binomial}(N_{d,n},p_d).
\]
Therefore, on $\{N_{d,n}>0\}$,
\[
E(\widehat p_{d,n}\mid D_1,\ldots,D_n)=p_d,
\qquad
\operatorname{Var}(\widehat p_{d,n}\mid D_1,\ldots,D_n)
=
\frac{p_d(1-p_d)}{N_{d,n}}.
\]

For every
$\varepsilon>0$ and every positive integer $L$, conditional Chebyshev's
inequality gives
\[
\Pr\{|\widehat p_{d,n}-p_d|>\varepsilon\}
\leq
\Pr(N_{d,n}<L)
+
\frac{p_d(1-p_d)}{\varepsilon^2L}.
\]
Since $N_{d,n} / n \overset{p}{\longrightarrow} \rho_d > 0$, $\Pr(N_{d,n} < L) \longrightarrow 0$ for every fixed $L$. Taking the upper limit in $n$ and then letting $L\to\infty$ proves
$\widehat p_{d,n}\overset{p}{\longrightarrow} p_d$.  Positivity and monotone selection imply
$p_1\geq p_0>0$. Thus, the event $\{\widehat p_{1,n}=0\}$ has
probability tending to zero, and the continuous mapping theorem gives
\[
\widehat\pi_n
\overset{p}{\longrightarrow}
\min\{1,p_0/p_1\}
=p_0/p_1
=\pi.
\]
Finally,
\[
\frac{M_{1,n}}{n}
=
\frac{N_{1,n}}{n}\widehat p_{1,n}
\overset{p}{\longrightarrow}\rho_1p_1.
\]
Combining this result with $m_n/M_{1,n}\overset{p}{\longrightarrow}\lambda$ gives $m_n/n\overset{p}{\longrightarrow}\lambda\rho_1p_1>0$. Hence $m_n\overset{p}{\longrightarrow}\infty$.
\end{proof}

\begin{proof}[\textbf{Proof of Proposition \ref{prop:full-equivalence}}]
Suppose first that $Q=H_XQ_x\in\mathcal Q_{\rm full}$. The chain rule for Radon--Nikodym derivatives and disintegrations gives
\[
\frac{dQ}{dP}(x,y)
=\frac{dH_X}{dP_X}(x)\frac{dQ_x}{dP_x}(y)
=\frac{\pi(x)}\pi\frac{dQ_x}{dP_x}(y)
\leq\frac1\pi \quad P\text{-almost surely},
\]
using the version specified in the proposition and defining the product as zero on $\{dH_X/dP_X=0\}$. Hence $Q\in\mathcal Q(P,\pi)$ and $Q_X=H_X$.

Conversely, let $w=dQ/dP$, assume $0\leq w\leq1/\pi$, and suppose $Q_X=H_X$. Set
\[
r(x)=\int w(x,y)P_x(dy).
\]
Then 
\[
r= \frac{dQ_X}{dP_X}=\frac{dH_X}{dP_X}=\frac{\pi(x)}{\pi} \quad  P_X\text{-almost surely}.
\]
For $H_X$-almost every $x$, $r(x)>0$, and disintegration gives
\[
Q_x(dy)=\frac{w(x,y)}{r(x)}P_x(dy).
\]
Therefore
\[
\frac{dQ_x}{dP_x}(y)
=\frac{\pi}{\pi(x)}w(x,y)
\leq\frac1{\pi(x)} \quad P_x\text{-almost surely}.
\]
This proves the reverse inclusion.
\end{proof}

\begin{proof}[\textbf{Proof of Proposition \ref{prop:full-likelihood}}]
Unit-level random assignment implied by Assumption \ref{ass:random} gives
\[
\mathcal F(X\mid D=0,S=1)
=\mathcal F\{X\mid S(0)=1\}.
\]
Under monotone selection, $S(0)=1$ if and only if the unit belongs to $\mathcal A$, almost surely. Therefore,
\[
\mathcal F(X\mid\mathcal A)
=\mathcal F(X\mid D=0,S=1)
=H_X.
\]
Thus, the $X$ marginal of $Q_0$ is $H_X$.

Fix $x$ with $s_0(x)>0$. Such covariate values have full $H_X$ measure. Monotone selection gives $s_1(x)>0$, and
\[
\Pr\{\mathcal A\mid S(1)=1,X=x\}
=\frac{s_0(x)}{s_1(x)}
=\pi(x).
\]
Conditional on $X=x$ and $S(1)=1$, the population contains only the always-selected and marginal-in strata. Hence,
\[
P_x=\pi(x)Q_{0,x}+\{1-\pi(x)\}R_{0,x},
\]
where $R_{0,x}=\mathcal F\{Y(1)\mid\mathcal M,X=x\}$ when the marginal-in stratum has positive conditional probability. Otherwise, choose $R_{0,x}$ arbitrarily.

For every measurable $B\subseteq\mathcal Y$,
\[
\pi(x)Q_{0,x}(B)\leq P_x(B).
\]
Since $\pi(x)>0$, this measure inequality implies
\[
Q_{0,x}\ll P_x,
\qquad
0\leq\frac{dQ_{0,x}}{dP_x}\leq\frac{1}{\pi(x)}
\quad P_x\text{-almost surely}.
\]
Thus, $Q_0\in\mathcal Q_{\rm full}$. If $\pi(x)<1$, the residual distribution is
\[
R_{0,x}(B)
=\frac{P_x(B)-\pi(x)Q_{0,x}(B)}{1-\pi(x)}.
\]
If $\pi(x)=1$, the density bound becomes $0\leq dQ_{0,x}/dP_x\leq1$. Since both measures have total mass one, the density equals one $P_x$-almost surely, and $Q_{0,x}=P_x$.
\end{proof}

\begin{proof}[\textbf{Proof of Proposition \ref{prop:full-sharpness}}]
Fix $\widetilde Q\in\mathcal Q_{\rm full}$ and choose measurable versions of the kernels $\widetilde Q_x$, $P_x$, and $G_x$. We construct a compatible data-generating process.

First, draw $X\sim F_X$. Conditional on $X=x$, draw $T\in\{\mathcal A,\mathcal M,\mathcal N\}$ with probabilities
\[
\begin{aligned}
\Pr(T=\mathcal A\mid X=x)&=s_0(x),\\
\Pr(T=\mathcal M\mid X=x)&=s_1(x)-s_0(x),\\
\Pr(T=\mathcal N\mid X=x)&=1-s_1(x).
\end{aligned}
\]
These probabilities are nonnegative and sum to one. Define
\[
(S(0),S(1))=
\begin{cases}
(1,1), & T=\mathcal A,\\
(0,1), & T=\mathcal M,\\
(0,0), & T=\mathcal N.
\end{cases}
\]
Then monotone selection holds and
\[
\Pr\{S(d)=1\mid X=x\}=s_d(x),
\qquad d\in\{0,1\}.
\]

Next, define the conditional distribution of $Y(1)$. For $T=\mathcal A$, set
\[
Y(1)\mid X=x,T=\mathcal A\sim\widetilde Q_x.
\]
For $T=\mathcal M$ and $0<\pi(x)<1$, define
\[
\widetilde R_x(B)
=\frac{P_x(B)-\pi(x)\widetilde Q_x(B)}{1-\pi(x)}.
\]
The conditional density restriction in $\mathcal Q_{\rm full}$ makes the numerator a nonnegative measure of total mass $1-\pi(x)$, so $\widetilde R_x$ is a probability distribution. If $\pi(x)=0$, set $\widetilde R_x=P_x$. If $\pi(x)=1$, the marginal-in stratum has zero conditional probability; choose $\widetilde R_x$ arbitrarily. For $H_X$-almost every $x$ with $\pi(x)=1$, the density restriction also gives $\widetilde Q_x=P_x$. For $T=\mathcal N$, define $Y(1)$ arbitrarily.

For $T=\mathcal A$, choose a coupling of $Y(0)$ and $Y(1)$ with marginals $G_x$ and $\widetilde Q_x$. One admissible choice is conditional independence:
\[
(Y(0),Y(1))\mid X=x,T=\mathcal A
\sim
G_x\otimes\widetilde Q_x.
\]
For $T\in\{\mathcal M,\mathcal N\}$, define $Y(0)$ arbitrarily because these units are not selected under control.

Finally, for a generic unit, append a binary assignment variable independently of the constructed latent variables. For a sample, take independent copies of the latent construction and append the prescribed joint distribution of $\mathbf D_n$ independently of those copies. 

We verify the observed distribution. Consistency and unit-level random assignment implied by Assumption \ref{ass:random} give

\[
\Pr(S=1\mid D=d,X=x)=s_d(x),
\qquad d\in\{0,1\}.
\]
For $x$ with $s_1(x)>0$,
\[
\begin{aligned}
\mathcal F(Y\mid D=1,S=1,X=x)
&=\mathcal F\{Y(1)\mid S(1)=1,X=x\}\\
&=\pi(x)\widetilde Q_x+\{1-\pi(x)\}\widetilde R_x\\
&=P_x.
\end{aligned}
\]
For $x$ with $s_0(x)>0$, selected controls satisfy $T=\mathcal A$, so
\[
\mathcal F(Y\mid D=0,S=1,X=x)
=\mathcal F\{Y(0)\mid\mathcal A,X=x\}
=G_x.
\]
The construction also draws $X$ from $F_X$, so it matches every component of the full observed distribution.

It remains to verify the target. Since $T=\mathcal A$ with conditional probability $s_0(x)$,
\[
\begin{aligned}
\mathcal F\{X,Y(1)\mid\mathcal A\}(dx,dy)
&=\frac{s_0(x)}{p_0}F_X(dx)\widetilde Q_x(dy)\\
&=H_X(dx)\widetilde Q_x(dy)\\
&=\widetilde Q(dx,dy).
\end{aligned}
\]
Thus, every element of $\mathcal Q_{\rm full}$ is attainable while the full observed distribution remains fixed. Proposition \ref{prop:full-likelihood} gives the reverse inclusion. Hence, $\mathcal Q_{\rm full}$ is the sharp identification region.
\end{proof}

\newpage
\bibliographystyle{unsrtnat}
\bibliography{CLI_references}
\end{document}